\documentclass[a4paper,11pt,oneside]{article}       % onecolumn (second format)
\usepackage{graphicx}
\usepackage[english]{babel}
\usepackage[mathscr]{eucal}
\usepackage{amsmath}
\usepackage{mathrsfs}
\usepackage{amsthm}
\usepackage{amsfonts}
\usepackage{amssymb}
\usepackage{amscd}
\usepackage{bbm}
\usepackage[affil-it]{authblk}

\newcommand{\ud}{\mathrm{d}}
\newcommand{\ii}{\mathrm{i}}

\newcommand{\ind}[1]{#1_{N_1,N_2}^{(k_1,k_2)}}

\newcommand{\bra}[1]{\langle #1|}
\newcommand{\ket}[1]{|#1\rangle}

\newcommand{\tr}{\mathrm{Tr}}

\newcommand{\de}{\partial}
\newcommand{\alf}[2]{\alpha^{(#1,#2)}}

\newtheorem{theorem}{Theorem}[section]

\newtheorem{lemma}{Lemma}[section]
\newtheorem{proposition}{Proposition}[section]
\newtheorem{corollary}{Corollary}[section]
\numberwithin{equation}{section}
\newcommand{\cH}{\mathcal{H}}

\begin{document}

\title{Mean-field quantum dynamics for a mixture of Bose-Einstein condensates%\thanks{This work is partially supported by the 2014-2017 MIUR-FIR grant ``\emph{Cond-Math: Condensed Matter and Mathematical Physics}'', code RBFR13WAET, and by a 2015 visiting research fellowship at the International Center for Mathematical Research CIRM, Trento}
}

%\titlerunning{Mean-field quantum dynamics for a mixture of Bose-Einstein condensates} % if too long for running head

\author{Alessandro Michelangeli \thanks{michel@math.lmu.de, alemiche@sissa.it}}
\affil{SISSA -- International School for Advanced Studies \\
				Via Bonomea 265, 34136 Trieste (Italy)}
\author{Alessandro Olgiati \thanks{aolgiati@sissa.it}}
\affil{   SISSA -- International School for Advanced Studies \\
				Via Bonomea 265, 34136 Trieste (Italy) \\
             }  


\maketitle

\begin{abstract}
We study the effective time evolution of a large quantum system consisting of a mixture of different species of identical bosons in interaction. If the system is initially prepared so as to exhibit condensation in each component, we prove that condensation persists at later times and we show quantitatively that the many-body Schr\"{o}dinger dynamics is effectively described by a system of coupled cubic non-linear Schr\"{o}dinger equations, one for each component.\vspace{0.5cm}

\end{abstract}

%%%%%  SECTION 1

\section{Introduction: BEC mixtures and non-linear effective dynamics}\label{sec:intro}

Bose-Einstein condensation is the well-known quantum phenomenon that occurs in a large many-body system of identical bosons when a macroscopic number of particles occupy the same one-body state. It can be thought of as a crowding phenomenon where the many-body state is for a large part factorised into the product of identical copies of the same one-body orbital, the condensate wave-function.

The mathematics of the Bose gas and its condensation is an extraordinarily rich and active subject that dates back from the very first systematic treatment of Bogolubov in the mid 1940's (that is, some 20 years after the theoretical discovery of BEC), boosted recently by the first experimental realisations of condensates in the mid 1990's and the subsequent advances in the techniques for manipulating ultra-cold atoms. The literature is therefore huge: we refer to the monograph \cite{LSeSY-ober} for what concerns the ``static'' picture (emergence of BEC, characterisation of the ground state, etc.) and to the works we will be mentioning in the following and to the references therein for the ``dynamical'' picture (the condensate's evolution).

Let us only recall here a few main and well-known points in order to formulate our problem. Let us consider the Hilbert space 
\begin{equation}
\mathfrak{h}\;=\;L^2(\Lambda)\otimes \mathbb{C}^{2s+1}
\end{equation}
for a quantum particle of spin $s\in\frac{1}{2}\mathbb{N}_0$ confined in a domain $\Lambda\subset\mathbb{R}^d$ (possibly $\mathbb{R}^d$ itself) and, correspondingly, for a  $d$-dimensional (bosonic) system of $N$ such particles, let us consider the Hilbert spaces 
\begin{equation}\label{eq:Hilbert_N}
\cH_N\;=\;\mathfrak{h}^{\otimes N}\,,\qquad\cH_{N,\mathrm{sym}}\;=\;\mathfrak{h}^{\otimes N}_\mathrm{sym}\,,
\end{equation}
namely the $N$-fold and the  \emph{symmetric} $N$-fold tensor product of $\mathfrak{h}$. 
Let $\gamma_N$, a positive trace-class operator on $\cH_{N,\mathrm{sym}}$ with unit trace, be the density matrix describing a state of a given bosonic system. Consistently with the physical notion of ``occupation numbers'', the standard mathematical tool to express the occurrence of BEC when the system is in the state $\gamma_N$  is the so-called one-body marginal (or one-body reduced density matrix)
\begin{equation}\label{eq:partial_trace1}
\gamma_N^{(1)}\;=\;\mathrm{Tr}_{N-1}\,\gamma_N\,.
\end{equation}
Here, the map $\mathrm{Tr}_{N-1}:\mathcal{B}_1(\cH_{N,\mathrm{sym}})\to\mathcal{B}_1(\mathfrak{h})$ is the \emph{partial trace} from trace class operators on $\cH_N$ to trace class operators on $\mathfrak{h}$, defined by
\begin{equation}\label{eq:def_partial_trace_with_basis}
\langle \varphi,(\mathrm{Tr}_{N-1}\,T)\psi\rangle_\mathfrak{h}\;=\;\sum_{k}\langle\varphi\otimes\xi_k,T\,\psi\otimes\xi_k\rangle_{\cH_N}\qquad\forall\varphi,\psi\in\mathfrak{h},
\end{equation}
where $(\xi_k)_k$ is an orthonormal basis of $\cH_{N-1,\mathrm{sym}}$. Observe that \eqref{eq:def_partial_trace_with_basis} is in fact independent of the choice of the basis and is equivalent to
\begin{equation}\label{eq:def_partial_trace_without_basis}
 \mathrm{Tr}_\mathfrak{h}(A\cdot\mathrm{Tr}_{N-1}\,T)\;=\;\mathrm{Tr}_{\cH}(A\otimes\mathbbm{1}_{N-1})\cdot T))\qquad\forall A\in\mathcal{B}(\mathfrak{h})\,.
\end{equation}
Thus, $\gamma_N^{(1)}$ is obtained by ``tracing out'' $N-1$ degrees of freedom from $\gamma_N$: for example, for a system of $N$ spinless ($s=0$) bosons in the pure state $\Psi_{N,\mathrm{sym}}\in L^2(\mathbb{R}^{Nd},\ud x_1\cdots\ud x_N)$ ($L^2_{N,\mathrm{sim}}\equiv$ the wave-functions that are symmetric under permutation of any two variables), the corresponding one-body marginal $\gamma_N^{(1)}$ has kernel
\begin{equation}
\gamma_N^{(1)}(x,x')\;=\;\int_{\mathbb{R}^{(N-1)d}}\!\!\!\!\!\!\Psi_N(x,x_2,\dots,x_N)\,\overline{\Psi_N(x',x_2,\dots,x_N)}\,\ud x_2\cdots\ud x_N\,.
\end{equation}

Being a density matrix, the one-body marginal $\gamma_N^{(1)}$ has a complete set of real non-negative eigenvalues that sum up to 1, and being it the partial trace of a many-body state $\gamma_N$, it is natural to think of these eigenvalues as the occupation numbers in $\gamma_N$, that is, each eigenvalue of $\gamma_N^{(1)}$ can be interpreted as the fraction of the $N$ particles that are in the same one-body state given by eigenvector associated with the considered eigenvalue.

In a sense to be specified in the given context, one therefore says that the many-body state $\gamma_N$ exhibits \emph{condensation in the state} $\varphi\in\mathfrak{h}$ ($\|\varphi\|=1$) if $\varphi$ is an eigenvector of $\gamma_N^{(1)}$ that belongs to a non-degenerate eigenvalue that is by far larger than all other eigenvalues, i.e., it is almost 1 while all other eigenvalues are almost zero. In other words, $\gamma_N^{(1)}\approx|\varphi\rangle\langle\varphi|$, the rank-one projection onto $\varphi$. 
This notion of condensation becomes conceptually well-posed and mathematically rigorous in the limit $N\to\infty$ -- a genuine thermodynamic limit, or some simpler prescription on $N\to\infty$ that mimics the thermodynamic limit. For the present discussion we only consider the case of \emph{complete} condensation, namely
\begin{equation}\label{eq:def_100BEC}
\lim_{N\to\infty}\gamma_{N}^{(1)}\;=\;|\varphi\rangle\langle\varphi|\qquad(\textrm{complete BEC})\,.
\end{equation}
Even if a priori the limit in \eqref{eq:def_100BEC} can be stated in several inequivalent operator topologies, from the trace norm to the weak operator topology, the bounds
\begin{equation}\label{eq:equivalent-BEC-control-1compontent}
1-\langle\varphi,\gamma_N^{(1)}\varphi\rangle\;\leqslant\;\mathrm{Tr}\big|\,\gamma_{N}^{(1)}-|\varphi\rangle\langle\varphi|\,\big|\;\leqslant\;2\sqrt{1-\langle\varphi,\gamma_N^{(1)}\varphi\rangle}
\end{equation}
(see Lemma \ref{lemma:intertwined-controls} below)
show that the occurrence of the convergence $\gamma_{N}^{(1)}\to|\varphi\rangle\langle\varphi|$ can be monitored  equivalently in any of them.

In \eqref{eq:def_100BEC}  $\varphi$ is customarily referred to as the \emph{condensate wave-function} and in the presence of condensation the diagonal $\gamma_N^{(1)}(x,x)$ of the one-body marginal becomes, for large $N$, a good approximation of the condensate profile $|\varphi(x)|^2$. While the limit  \eqref{eq:def_100BEC} is naturally interpreted as if the many-body state was almost completely factorised as $\varphi^{\otimes N}$, the closeness $\gamma_N^{(1)}\approx|\varphi\rangle\langle\varphi|$ is obviously much weaker than the actual closeness  $\Psi_N\approx\varphi^{\otimes N}$ in the norm of $\cH_N$. It can instead be argued that $\gamma_N^{(1)}\approx|\varphi\rangle\langle\varphi|$ implies a factorisation of the many-body state when an amount $k=O(N)$ of particles are considered: for the precise meaning of this control, as well as for equivalent characterisation of complete BEC, we refer to \cite{am_equivalentBEC} and \cite[Section 2]{kp-2009-cmp2010}, as well as to the following considerations in Section \ref{sec:indicators_of_convergence}.

Within this formalism, the \emph{dynamical problem} for the time evolution of a condensate is the problem of the persistence of BEC, in the form of the asymptotic condition \eqref{eq:def_100BEC}, along the evolution
$\gamma_N\mapsto \gamma_{N,t}=e^{-\ii t H_N}\gamma_N e^{\ii t H_N}$ governed by a given many-body Hamiltonian $H_N$, and thus of the rigorous derivation of the law $\varphi\mapsto\varphi_t$ that gives the condensate wave-function at later times. Whereas the solution to the Schr\"{o}dinger equation for a system of $N$ interacting particles is obviously out of reach when $N$ is large, both analytically and numerically, the language of the reduced density matrix boils down the dynamical problem to the level of one-body states, a major simplification in which one renounces to the complete knowledge of $\gamma_{N,t}$. The price is the replacement of the many-body \emph{linear} dynamics for $\gamma_{N,t}$ with an effective \emph{non-linear} dynamics for $\varphi_t$, the non-linearity being due to the inter-particle interaction and emerging in the form of a non-linear self-interaction term (a cubic non-linearity, for typical two-body interactions) in the non-linear Schr\"{o}dinger equation for $\varphi_t$. In shorts, the dynamical problem consists of the completion of the following diagram, assuming that the first line holds at time $t=0$:
\begin{equation}\label{scheme_for_marginals}
\begin{CD}
\Psi_N @>\scriptsize\textrm{partial trace}>>\gamma_N^{(1)} @>N\to\infty>> |\varphi\rangle\langle\varphi| \\
@ V\scriptsize\begin{array}{c} \textrm{many-body} \\ \textrm{\textbf{linear} dynamics}  \end{array} VV @V  VV               @VV\scriptsize\begin{array}{c}\textrm{\textbf{nonlinear}} \\ \textrm{Schr\"{o}dinger eq.} \end{array}V    \\
\Psi_{\! N,t} @>\scriptsize\textrm{\qquad\qquad\qquad\;}>>\gamma_{\! N,t}^{(1)} @>N\to\infty>> |\varphi_t\rangle\langle\varphi_t|
\end{CD}
\end{equation}

There is a vast literature on the rigorous derivation of non-linear Schr\"{o}dinger equations as the effective equations for the dynamics of a many-body Bose gas that at time $t=0$ displays the asymptotic factorisation \eqref{eq:def_100BEC}. It covers different space dimensions ($d=1,2,3$), a wide range of local singularities and long-distance decays for the inter-particle interactions, and various types of scaling limits in the many-body Hamiltonian $H_N$. We refer to the reviews \cite{S-2007,S-2008,Benedikter-Porta-Schlein-2015} for a comprehensive outlook, remarking  that this  problem  has involved a variety of approaches and techniques from analysis, operator theory, kinetic theory, and probability.

Let us emphasize that in the lack (so far) of a rigorous control of the asymptotics \eqref{eq:def_100BEC} in a genuine thermodynamic limit, it is customary to investigate and  reproduce it in an ad hoc scaling limit in which the $N$-body Hamiltonian is suitably re-scaled with $N$, thus making the actual inter-particle interaction $N$-dependent. This artifact on the one hand allows for an explicit determination of the limit $N\to\infty$ in $\gamma_N^{(1)}$, while on the other hand it preserves at any $N$ an amount of relevant physical features of the system and of its Hamiltonian, among which the property that kinetic and potential energy remain of the same order so that the interaction is still visible in the limit, as well as certain dilution properties of the system and short-range features of the interaction. Typical relevant scaling limits are those in which the two-body potential $V$ that models the interaction among particles is replaced by a $N$-dependent two-body potential
\begin{equation}\label{eq:scaling_limits}
V_N(x)\;=\;N^{3\beta-1}V(N^\beta x)\,,\qquad \beta\in[0,1]
\end{equation}
(here $x=x_i-x_j$ is the relative coordinate between particle $i$ and particle $j$). The regime $\beta=0$ is the mean-field regime, whereas $\beta=1$ gives the so-called Gross-Pitaevskii scaling regime. We refer to \cite{am_GPlim} for an extended discussion.

Coming now to the problem we intend to study, let us mention that a large amount of theoretical and experimental studies in the dynamics of Bose gases involve the interaction among two or more samples, each of which is in a condensate (see \cite[Section 12.11]{ps2003} and the references therein). These configurations are usually realised using atoms occupying different hyperfine states \cite{MBGCW-1997,Stamper-Kurn_Ketterke_et_al_PRL-1998}
or also different atoms \cite{Modugno-PRL-2002}.
In either case the particles of the two samples must be considered as different species and one then refers to such a system as a \emph{mixture} of condensates. For mixtures one has the additional possibility of producing transitions between the two components, typically by applying an external oscillating magnetic field tuned close to the hyperfine splitting, which gives rise to a new variety of physical phenomena.
In suitable circumstances condensation is robust enough to be preserved in time, each condensate keeping its own individuality, even in the presence of a significant transfer of atoms from one state to the other.

Physical arguments corroborated by experimental data show that to a very good approximation the effective dynamics of the mixture is described by a system of two coupled non-linear Schr\"{o}dinger equations, the coupling among them accounting for how each of the two condensates affect the evolution of the other. The rich variety of phases exhibited by the mixture is also well described in terms of the sign and the magnitude of the effective couplings appearing in the non-linear system \cite[Section 12.11]{ps2003}. In analogy to the rigorous derivation of the effective evolution equation from the quantum dynamics of a condensate, we are therefore interested in reproducing the \emph{effective evolution equations for the mixture dynamics}.

To this aim, we shall study a large three-dimensional system consisting of two distinguishable populations of identical bosons, with two-body interactions among particles of the same species and of different species. We also allow for the possibility of an external electro-magnetic field coupled with the particles, by means of a confining potential in case of a mixture in a trap or an external magnetic potential. We will assume that at time $t=0$ condensation occurs in both components and we will consider the quantum evolution of this mixture. 
As we shall discuss in Sections \ref{sec:model}, \ref{sec:indicators_of_convergence}, and \ref{sec:mean-field-scaling}, where the model will be set up, the appropriate language to formulate the assumption of condensation and to conveniently monitor the many-body dynamics is the generalisation of the notion of reduced density matrix for a Bose gas with two distinguishable components. The effective dynamics will then emerge in a scheme that doubles \eqref{scheme_for_marginals} in a suitable sense: this is our main result, presented in Section \ref{sec:model}.

On a more technical level, we find that several among the many alternative techniques developed so far for the (one-component) dynamical problem  \eqref{scheme_for_marginals} can be conveniently adapted to approach the multiple-component setting, with an amount of non-trivial modifications that depend on the considered techniques. In this work we employ a particularly robust and versatile method, invented and refined in a recent series of papers by P.~Pickl \cite{Pickl-JSP-2010,Pickl-LMP-2011,Pickl-RMP-2015}, with the contribution of A.~Knowles \cite{kp-2009-cmp2010}, which monitors how the displacement $\gamma_{N,t}^{(1)}-|\varphi_t\rangle\langle\varphi_t|$ changes with time by means of an ad hoc ``counting'' of the amount of particles in the many-body state $\gamma_{N,t}$ that occupy the one-body state $\varphi_t$. A number of tools and estimates that are crucial for this approach are reviewed in Appendix \ref{sec:particle-counting_method}. Such a counting method  has the virtue of being based on a few key algebraic steps that only involve the potential part of the Hamiltonian, and the estimates that then follow do not concern the differential part of the operator.
%%%%%%%%%%%%%%%%%
In Section  \ref{sec:main_proof}, where we prove our main result, we discuss the non-trivial adaptation that we developed for Pickl's method in the multi-component setting.

In order to access complementary aspects of the present analysis, in particular the very much relevant control of the fluctuations around the emergent effective dynamics, a parallel study by one of us in collaboration with G.~De Oliveira \cite{DeOliveira-Michelangeli-2016} is being developed within the Fock space approach \cite[Chapters 3 and 4]{Benedikter-Porta-Schlein-2015}, so that the two works are in fact part of a unified project.

One last comment is about our treatment of the problem within the \emph{mean-field} approximation. As our primary goal is to give evidence of the mechanism by which the many-body Schr\"{o}dinger dynamics of an initial mixture of condensates give rise to an effective dynamics of coupled non-linear equations for the persistence of condensation in each component, we found it instructive to place our discussion in the technically easiest framework, the mean field. As discussed in Section \ref{sec:mean-field-scaling}, in order to derive the  effective dynamics actually expected in a regime of almost zero temperature and high dilution one has to adopt  the more realistic Gross-Pitaevskii scaling. Given the high versatility of Pickl's method, that has been proved to be successful also for the Gross-Pitaevskii scaling \cite{Pickl-JSP-2010,Pickl-RMP-2015}, the results of our multi-component analysis too can be proved to hold  also for  in such a (technically more difficult) scheme.

%On the other hand, Pickl's method is somewhat rigid in the sense that it only applies to a very specific indicator of condensation, namely the quantity $\alpha_{N_1,N_2}^{(k_1,k_2)}$ -- see \eqref{eq:indicator-alpha} below. It does not give access to other relevant information such as the control of the fluctuations around the emergent effective dynamics: we point out that this complementary aspect of the present analysis is the object of a parallel study by one of us in collaboration with G.~De Oliveira \cite{DeOliveira-Michelangeli-2016}, so that the two works are in fact part of a unified project.

\textbf{Notation.} Essentially all the notation adopted here is standard, and we defer to Subsection \ref{subsec:notation_for_the_proof} the introduction of ad hoc extra notation for the proof of our main result. Let us only emphasize the following. As customary, by ``$\lesssim$'' we shall mean inequalities with a universal constant as an overall pre-factor, otherwise the dependence of the constants on other quantities of interest will be declared. Scalar products and norms in the considered Hilbert spaces will only be explicitly indexed whenever the underlying Hilbert space is not evident from the context or when we want to emphasize it. In the same spirit we shall deal with the identity operator $\mathbbm{1}$. We shall also adopt the customary convention to distinguish the \emph{operator} domain and the \emph{form} domain of any given self-adjoint operator $H$ by means of the notation $\mathcal{D}(H)$ vs $\mathcal{D}[H]$. When taking the trace of an operator, suitable subscripts will be occasionally added, such as $\mathrm{Tr}_{\mathfrak{h}}$ or $\mathrm{Tr}_{\mathfrak{h}^{\otimes 2}}$, in order to clarify the underlying Hilbert space, while unambiguously symbols like $\mathrm{Tr}_{N-k}$ will denote the partial trace map.

\section{Model and main results}\label{sec:model}

\subsection{Many-body and one-body picture and condensate mixture}\label{subsec:manybd-onebd-condensate-mixture}

For $N_1,N_2\in\mathbb{N}$ we consider the Hilbert spaces
\begin{equation}
\begin{split}
\cH_{N_1,N_2}\;&:=\;\cH_{N_1}\otimes\cH_{N_2} \\
&=\;L^2(\mathbb{R}^{3N_1},\ud x_1\cdots \ud x_{N_1})\otimes L^2(\mathbb{R}^{3N_2},\ud y_1\cdots \ud y_{N_2})\,,
\end{split}
\end{equation}
\begin{equation}
\begin{split}
\cH_{N_1,N_2,\mathrm{sym}}\;&:=\;\cH_{N_1,\mathrm{sym}}\otimes\cH_{N_2,\mathrm{sym}} \\
&=\;L^2_{\mathrm{sym}}(\mathbb{R}^{3N_1},\ud x_1\cdots \ud x_{N_1})\otimes L^2_{\mathrm{sym}}(\mathbb{R}^{3N_2},\ud y_1\cdots \ud y_{N_2})
\end{split}
\end{equation}
where $\cH_{N_1,N_2}$ (resp., $\cH_{N,\mathrm{sym}}$) is the $N$-body (resp. $N$-body bosonic) Hilbert space introduced in \eqref{eq:Hilbert_N} built upon the one-body Hilbert space $\mathfrak{h}=L^2(\mathbb{R}^3)$. An element $\Psi$ of $\cH_{N_1,N_2,\mathrm{sym}}$ is a square-integrable function with two distinguishable sets of variables and
\[
\Psi(x_1,\dots,x_{N_1};y_1,\dots,y_{N_2})
\]
is invariant under exchange of any two $x$-variables or any two $y$-variables, with no overall permutation symmetry among the two sets of variables. The same double bosonic symmetry is induced on the density matrices acting on $\cH_{N_1,N_2,\mathrm{sym}}$.

The states of $\cH_{N_1,N_2,\mathrm{sym}}$ describe systems consisting of $N_1$ identical bosons of the species A  and $N_2$ identical bosons of the (different) species B. We want to focus on those states where condensation occurs in both species. For concreteness one may think of the special case of a density matrix $\gamma_{N_1}\otimes\gamma_{N_2}$ on $\cH_{N_1,N_2,\mathrm{sym}}=\cH_{N_1,\mathrm{sym}}\otimes\cH_{N_2,\mathrm{sym}}$ where for $j=1,2$ $\gamma_{N_j,\mathrm{sym}}$ is a bosonic density matrix on $\cH_{N_j}$ that exhibits complete condensation in the asymptotic sense \eqref{eq:def_100BEC}. For a generic density matrix $\gamma_{N_1,N_2}$ on $\cH_{N_1,N_2,\mathrm{sym}}$ the assumption of double condensation has a natural formulation in terms of the double partial trace realised by doubling (i.e., tensoring) the map $\mathrm{Tr}_{N-1}$ discussed in \eqref{eq:partial_trace1}-\eqref{eq:def_partial_trace_without_basis}. This leads us to introduce the \emph{``double'' reduced density matrix}
\begin{equation}\label{eq:def_double_partial_trace-ABSTRACT}
\gamma_{N_1,N_2}^{(1,1)}\;=\;\mathrm{Tr}_{N_1-1}\otimes\mathrm{Tr}_{N_2-1}\:\gamma_{N_1,N_2}
\end{equation}
associated with $\gamma_{N_1,N_2}$, clearly a density matrix on $\mathfrak{h}\otimes\mathfrak{h}=L^2(\mathbb{R}^3,\ud x)\otimes L^2(\mathbb{R}^3,\ud y)$. The operation in \eqref{eq:def_double_partial_trace-ABSTRACT} amounts to tracing out $N_1-1$ degrees of freedom of type A and $N_2-1$ degrees of freedom of type B from $\gamma_{N_1,N_2}$.
Explicitly, for a pure state (i.e, a normalised function) $\Psi_{N_1,N_2}\in\cH_{N_1,N_2,\mathrm{sym}}$ the associated $\gamma_{N_1,N_2}^{(1,1)}$ has integral kernel
\begin{equation}\label{eq:def_double_partial_trace-KERNEL}
\begin{split}
\gamma_{N_1,N_2}^{(1,1)}(x,x';y,y')\;=\;&\int_{\mathbb{R}^{3(N_1-1)}} \int_{\mathbb{R}^{3(N_2-1)}} \ud x_2\cdots\ud x_{N_1}\ud y_2\cdots\ud y_{N_2} \\
& \qquad\times \Psi_{N_1,N_2}(x,x_2,\dots,x_{N_1};y,y_2,\dots,y_{N_2}) \\
& \qquad\times \overline{\Psi_{N_1,N_2}}(x',x_2,\dots,x_{N_1};y',y_2,\dots,y_{N_2})\,.
\end{split}
\end{equation}
In general $\gamma_{N_1,N_2}^{(1,1)}$ is neither factorised as a product of two density matrices on $\mathfrak{h}\otimes\mathfrak{h}$ nor with rank one.

Instead of averaging out all but one particles for each species, one could also think of controlling a small portion of each component of the system that corresponds to an arbitrary number $k_1\leqslant N_1$ of A-particles and $k_2\leqslant N_2$ of B-particles, thus tracing out $N_j-k_j$ degrees of freedom, $j=1,2$, from $\gamma_{N_1,N_2}$. Definitions \eqref{eq:def_double_partial_trace-ABSTRACT}-\eqref{eq:def_double_partial_trace-KERNEL} are then modified straightforwardly so as to define the $(k_1,k_2)$-reduced density matrix $\gamma_{N_1,N_2}^{(k_1,k_2)}$ on the space $\cH_{k_1,k_2}$. This is the appropriate marginal to study  the particle correlations. Another relevant indicator is obtained by tracing out from the many-body state all the degrees of freedom of one component, and all but one of the other component, thus ending up with  $\gamma_{N_1,N_2}^{(1,0)}$ and $\gamma_{N_1,N_2}^{(0,1)}$, that are density matrices on the one-body space $\mathfrak{h}$. In Section \ref{sec:indicators_of_convergence} we develop a more systematic discussion on the marginals $\gamma_{N_1,N_2}^{(k_1,k_2)}$, their algebra, and an amount of useful bounds for them.

In full analogy with the corresponding definition for a one-component condensate, one says that the state $\gamma_{N_1,N_2}$ on  $\cH_{N_1,N_2,\mathrm{sym}}$ exhibits condensation for both species of particles (that is, a two-component mixture BEC), with condensate functions $u$ and $v$ (two normalised one-body wave-functions), if for the associated reduced density matrix one has
\begin{equation}\label{eq:def_100BEC-2component}
\lim_{\substack{N_1\to\infty \\ N_2\to\infty}}\gamma_{N_1,N_2}^{(1,1)}\;=\;|u\otimes v\rangle\langle u\otimes v|\;=\;|u\rangle\langle u|\otimes |v\rangle\langle v|\,.
\end{equation}
As for \eqref{eq:def_100BEC}, the limit in \eqref{eq:def_100BEC-2component} is of thermodynamic type. It expresses, in the interpretation of occupation numbers discussed in Section \ref{sec:intro}, the idea that the actual many-body state has the double-condensate form $u^{\otimes N_1}\otimes v^{\otimes N_2}$, although the vanishing of $\gamma_{N_1,N_2}^{(1,1)}-|u\otimes v\rangle\langle u\otimes v|$ is \emph{much weaker} than the actual vanishing of $\|\Psi_{N_1,N_2}-u^{\otimes N_1}\otimes v^{\otimes N_2}\|$.

Observe that the distinguishability of the two species results in a precise ordering in the product $u\otimes v\in\mathfrak{h}\otimes\mathfrak{h}$ of the two condensate functions. Thus, even in the case $u=v$ (as elements in $L^2(\mathbb{R}^3)$), the double condensation $\gamma_{N_1,N_2}^{(1,1)}\approx|u\otimes u\rangle\langle u\otimes u|$ expresses the fact that each component undergoes BEC with \emph{the same spatial profile} of the condensate: the two condensates then sit on top of each other, while the two species remain distinguishable.

While $\gamma_{N_1,N_2}^{(1,1)}$ allows for a simultaneous control of BEC for each species, the reduced density matrices $\gamma_{N_1,N_2}^{(1,0)}$ and $\gamma_{N_1,N_2}^{(0,1)}$ monitor the occurrence of condensation in one component, irrespectively of the other. 
In Section \ref{sec:indicators_of_convergence} (Lemma \ref{lemma:controllo}), we establish the bound
\begin{equation}\label{eq:bound_g11_g10_g01}
\begin{split}
\max\big\{1-\langle u,&\gamma_{N_1,N_2}^{(1,0)} u\rangle\,,\,1-\langle v,\gamma_{N_1,N_2}^{(0,1)} v\rangle\big\}\;\leqslant\;1-\langle u\otimes v,\gamma_{N_1,N_2}^{(1,1)} u\otimes v\rangle \\
&\leqslant\;(1-\langle u,\gamma_{N_1,N_2}^{(1,0)} u\rangle) + (1-\langle v,\gamma_{N_1,N_2}^{(0,1)} v\rangle)
\end{split}
\end{equation}
which shows that $\gamma_{N_1,N_2}^{(1,1)}\to |u\otimes v\rangle\langle u\otimes v|$ is equivalent to $\gamma_{N_1,N_2}^{(1,0)}\to|u\rangle\langle u|$ \emph{and} $\gamma_{N_1,N_2}^{(0,1)}\to|v\rangle\langle v|$.

\subsection{Double-component Hamiltonian}

Let us now come to the dynamical model we intend to study. First of all, we introduce two one-particle Hamiltonians $h_1$ and $h_2$ on $\mathfrak{h}$, one for each species. We have in mind for concreteness two non-relativistic Schr\"{o}dinger operators with given external magnetic and electric potentials, namely
\begin{equation}
h_1\;=\;-(\nabla_x-\ii A_1(x))^2+U_1(x)\,,\quad h_2\;=\;-(\nabla_y-\ii A_2(y))^2+U_2(y)\,,
\end{equation}
for suitable  measurable functions $A_j:\mathbb{R}^3\to\mathbb{R}^3$ (magnetic potentials)  and  $U_j:\mathbb{R}^3\to\mathbb{R}$ (trapping potentials), $j=1,2$, or their pseudo-relativistic version
\[
h_j\;=\;\sqrt{1-(\nabla_x-\ii A_j(x))^2}\,+\,U_j(x)\,,\qquad j=1,2\,.
\]
We then consider a model for $N_1$ identical bosons of one species and $N_2$ identical bosons of another species, in which each particle of the first (resp., of the second) species is subject to the one-body Hamiltonian $h_1$ (resp., $h_2$) and is coupled with the other particles of the same species via a two-body potential $V_1:\mathbb{R}^3\to\mathbb{R}$ (resp., $V_2:\mathbb{R}^3\to\mathbb{R}$), plus an additional inter-species two-body potential $V_{12}:\mathbb{R}^3\to\mathbb{R}$ that couples the two components of the mixture.

We  consider  the mean-field Hamiltonian for this system, namely the operator
\begin{equation}\label{eq:HN1N2}
\begin{split}
H_{N_1,N_2}\;=\;&\sum_{i=1}^{N_1}(h_1)_i^A+\frac{1}{N_1}\sum_{i<j}^{N_1}V_1(x_i-x_j) \\
+\;&\sum_{r=1}^{N_2}(h_2)_r^B+\frac{1}{N_2}\sum_{r<s}^{N_2}V_2(y_r-y_s) \\
&\quad +\frac{1}{N_1+N_2}\sum_{i=1}^{N_1}\sum_{r=1}^{N_2}V_{12}(x_i-y_r)
\end{split}
\end{equation}
acting on $\cH_{N_1,N_2}$. In \eqref{eq:HN1N2} and throughout this work, we adopt the following compact notation: if $T$ is an operator acting only on one of the two factors of $\cH_{N_1,N_2}$ and we need to consider it as an operator on the whole $\cH_{N_1,N_2}$ with trivial action on the other factor, we shall denote $T\otimes\mathbbm{1}$ (resp., $\mathbbm{1}\otimes T$) as $T^A$ (resp., $T^B$). Clearly $T^A$ and $S^B$ commute for any single-sector operators $T$ and $S$. Thus, if $V_{12}\equiv 0$, then $H_{N_1,N_2}$ consists of the sum of two (commuting) Hamiltonians, one for each species.

As we intend to study the quantum evolution governed by $H_{N_1,N_2}$ in the limit of very large $N_1$ and $N_2$, keeping the ratio $N_1/N_2$ (asymptotically) constant and non-zero, it is easily seen that the mean-field pre-factors $N_1^{-1}$, $N_2^{-1}$, and $(N_1+N_2)^{-1}$ inserted in front of the potential terms ensure that in this limit the kinetic and the potential part of the Hamiltonian remain comparable. Indeed, there are $N_1+N_2$ kinetic terms and $\frac{1}{2}(N_1+N_2)(N_1+N_2-1)$ potential terms in $H_{N_1,N_2}$, however the mean-field pre-factors reduce the order of the potential energy to
\[
\frac{1}{N_1}\cdot\frac{N_1(N_1-1)}{2}+\frac{1}{N_2}\cdot\frac{N_2(N_2-1)}{2}+\frac{1}{N_1+N_2}\cdot N_1\, N_2 \;=\;O(N_1+N_2)\,.
\]
There are of course \emph{other choices} that would preserve the mean-field character of the Hamiltonian, for example a common pre-factor $(N_1+N_2)^{-1}$ in front of all potential terms, or a pre-factor $(N_1N_2)^{-1/2}$ in front of the mixed interaction term. In Section \ref{sec:mean-field-scaling} we will show that our choice in  \eqref{eq:HN1N2} is \emph{the} physically meaningful one, for it yields the physically correct effective dynamics.

We also remark that $H_{N_1,N_2}$, albeit not factorised (unless $V_{12}\equiv 0$), maps $\cH_{N_1,N_2,\mathrm{sym}}$ into itself, because of its permutation symmetry  separately in each set of variables. Therefore, given a density matrix $\gamma_{N_1,N_2}$ on $\cH_{N_1,N_2,\mathrm{sym}}$, i.e., a many-body state that is bosonic for each species of particles, its evolution $e^{-\ii t H_{N_1,N_2}}\gamma_{N_1,N_2}e^{\ii t H_{N_1,N_2}}$ along the dynamics generated by $H_{N_1,N_2}$ remains a density matrix on $\cH_{N_1,N_2,\mathrm{sym}}$. In this evolution the number of particles of each kind is constant, there is an inter-species interaction ($V_{12}$) but no transfer of particles among the two populations.

\subsection{Asymptotic limit and effective dynamics}

As discussed in the Introduction, if the system is prepared at time $t=0$ in a state $\gamma_{N_1,N_2}$ where (complete) BEC occurs in both components, one expects that
\begin{itemize}
 \item condensation persists also at later times for the quantum dynamics generated by $H_{N_1,N_2}$, 
 \item and the evolution of the two condensate functions is governed by a system of two coupled non-linear Schr\"{o}dinger equations.
\end{itemize}
 In this work we prove these facts in the sense of the reduced density matrices.

What equations have to be expected can be seen by means of several heuristic arguments. We discuss them in detail in Section \ref{sec:mean-field-scaling}. To give a sketch here of the  `formal' derivation, let us fix first of all the precise sense of the double limit $N_1\to\infty$, $N_2\to\infty$. We want the two population numbers to remain comparable when they become arbitrarily large, so we assume that there are two constants $c_1,c_2>0$ such that
\begin{equation}\label{eq:N1N2ratio}
\lim_{\substack{N_1\to\infty \\ N_2\to\infty}}\,\frac{N_j}{N_1+N_2}\;=\;c_j \qquad (j=1,2)\,,
\end{equation}
that is, the ratio $N_1/N_2$ is assumed to be asymptotically constant. From now on, by $\lim_{N_1,N_2\to\infty}$ we shall mean a limit under the constraint \eqref{eq:N1N2ratio}.

To see formally the emergence of the effective dynamics it is enough to make the Ansatz
\[
%\Psi_{N_1,N_2,t}\;=\;
\lim_{N_1,N_2\to\infty}\,e^{-\ii t H_{N_1,N_2}}(u_0^{\otimes N_1}\otimes v_0^{\otimes N_2})\;=\;u_t^{\otimes N_1}\otimes v_t^{\otimes N_2}
\]
for some time-dependent functions $u_t$ and $v_t$ in $\mathfrak{h}$, given $u_0$ and $v_0$ at time $t=0$. This, of course, cannot be true at finite $N_1$, $N_2$, owing to the presence of the mixed interaction potential $V_{12}$ in $H_{N_1,N_2}$, still, one can compute the limit energy-per-particle functional
\[
\mathcal{E}[u_t,v_t]\;=\;\lim_{N_1,N_2\to\infty}\Big\langle u_t^{\otimes N_1}\otimes v_t^{\otimes N_2}\,,\,\frac{H_{N_1,N_2}}{N_1+N_2}\,u_t^{\otimes N_1}\otimes v_t^{\otimes N_2} \Big\rangle
\]
and then determine, by taking the variations of $\mathcal{E}[u_t,v_t]$, the equations that $u_t$ and $v_t$ must satisfy in order to preserve the value of this energy in time. Alternatively, one can re-write the initial value problem for many-body Schr\"{o}dinger equation
\[
\begin{cases}
\;\ii\partial_t\Psi_{N_1,N_2,t}\;=\;H_{N_1,N_2}\Psi_{N_1,N_2,t} \\
\;\Psi_{N_1,N_2,t}\big|_{t=0}\;=\;u_0^{\otimes N_1}\otimes v_0^{\otimes N_2}
\end{cases}
\]
in terms of the corresponding initial value problem for the finite hierarchy of equations for the marginals $\gamma_{N_1,N_2,t}^{(k_1,k_2)}$ (the BBGKY hierarchy) and close formally the limit hierarchy as $N_1,N_2\to\infty$ by plugging in the Ansatz above, re-written in the reduced density matrix form
\[
\lim_{N_1,N_2\to\infty}\,\gamma_{N_1,N_2,t}^{(k_1,k_2)}\;=\;|u_t^{\otimes k_1}\rangle\langle u_t^{\otimes k_1}|\otimes|v_t^{\otimes k_2}\rangle\langle v_t^{\otimes k_2}|\,;
\]
this decouples the equation for $\gamma^{(1,1)}_{\infty,\infty}\equiv\lim_{N_1,N_2\to\infty}\gamma_{N_1,N_2,t}^{(1,1)}$, which yields the corresponding equations for $u_t$ and $v_t$. We work out these formal calculations  in Section \ref{sec:mean-field-scaling}.

What one finds is the system
\begin{equation}\label{eq:Hartree_system}
\begin{split}
\ii\partial_t u_t\;&=\;h_1 u_t + (V_1*|u_t|^2) u_t + c_2 (V_{12}*|v_t|^2) u_t \\
\ii\partial_t v_t\;&=\;h_2 v_t + (V_2*|v_t|^2) v_t + c_1 (V_{12}*|u_t|^2) v_t
\end{split}
\end{equation}
of two coupled non-linear Schr\"{o}dinger equations of Hartree type, in the two unknowns $u_t$ and $v_t$, with initial condition $u_{t=0}=u_0$ and $v_{t=0}=v_0$. Observe in the first equation of \eqref{eq:Hartree_system} the two self-interaction terms: the first, $(V_1*|u_t|^2) u_t$, accounts for the interaction of an A-particle with the effective potential due to the presence of the other particles of the same species around it, the second, $c_2 (V_{12}*|v_t|^2) u_t$ accounts for the analog effective interaction of the A-particle with the B-particles. Same considerations for the second equation in \eqref{eq:Hartree_system}. The weights $c_1$ and $c_2$ adjust the magnitude of each inter-species self-interaction term with respect to the the relative ratio of each population in terms of the total number of particles.

It can be argued that, within the mean-field scheme, \eqref{eq:Hartree_system} is the correct system of evolution equations for the effective dynamics of a condensate mixture. Indeed, it is the mean-field version of the system of non-linear Schr\"{o}dinger equations that physical theoretical heuristics produce, in extraordinary agreement with the experimental data, when the mixture is modelled with the more realistic assumption of strong short-scale interactions and high dilution. This is another point that we discuss in detail in in Section \ref{sec:mean-field-scaling}.

\subsection{Assumptions and main theorem}

We shall work under the following set of assumptions.

\begin{itemize}
 \item[(A1)] The one-particle Hamiltonians $h_1$ and $h_2$ are self-adjoint and semi-bounded below on $\mathfrak{h}$. It is not restrictive to assume both of them positive. This implies that for any $N_1,N_2\in\mathbb{N}$ the free (kinetic) part in the Hamiltonian \eqref{eq:HN1N2}, namely
 \begin{equation}
H_{N_1,N_2}^{(0)}\;:=\;\sum_{i=1}^{N_1}(h_1)_i^A + \sum_{r=1}^{N_2}(h_2)_r^B\,,
 \end{equation}
 is self-adjoint and positive on $\cH_{N_1,N_2}$, and that the corresponding form domain $\mathcal{D}[H_{N_1,N_2}^{(0)}]$ is a Hilbert space w.r.t.~the scalar product
 \begin{equation}
 \langle \Psi,\Phi\rangle_{\mathcal{D}[H_{N_1,N_2}^{(0)}]}\;=\;\langle (\mathbbm{1}+H_{N_1,N_2}^{(0)})^{1/2}\Psi,(\mathbbm{1}+H_{N_1,N_2}^{(0)})^{1/2}\Phi\rangle\,.
 \end{equation}
  \item[(A2)] The potentials $V_1$, $V_2$, and $V_{12}$ are real-valued even functions satisfying
 \begin{equation}
 \begin{array}{c}
V_\alpha\in L^{r_\alpha}(\mathbb{R}^3)+L^{s_\alpha}(\mathbb{R}^3) \\
\textrm{for some }2\leqslant r_\alpha\leqslant s_\alpha\leqslant +\infty
 \end{array}\,,\qquad \alpha\in\{1,2,12\}\,.
 \end{equation}
 \item[(A3)] For all $N_1,N_2\in\mathbb{N}$ the Hamiltonian $H_{N_1,N_2}$ \eqref{eq:HN1N2} is self-adjoint and semi-bounded below on $\cH_{N_1,N_2}$, and $\mathcal{D}[H_{N_1,N_2}]\subset\mathcal{D}[H_{N_1,N_2}^{(0)}]$.
 \item[(A4)] The initial value problem consisting of the system \eqref{eq:Hartree_system} with initial condition $u(0)=u_0$ and $v(0)=v_0$ for given functions $u_0\in\mathcal{D}[h_1]$ and $v_0\in\mathcal{D}[h_2]$
  has a unique global-in-time solution
 \begin{equation}\label{soluzioni}
 (u,v)\;\in\;C(\mathbb{R},X)\cap C^1(\mathbb{R},\mathcal{D}[h_1]^*\oplus\mathcal{D}[h_2]^*)
 \end{equation}
 where
 \begin{equation}\label{eq:energy_space_1body}
 X\;:=\;(\mathcal{D}[h_1]\cap L^{\max\{\widehat{r}_1,\widehat{r}_{12}\}}(\mathbb{R}^3))\oplus(\mathcal{D}[h_2]\cap L^{\max\{\widehat{r}_2,\widehat{r}_{12}\}}(\mathbb{R}^3))
 \end{equation}
 and
 \begin{equation}
 \frac{1}{r_\alpha}+\frac{1}{\:\widehat{r}_\alpha}\;=\;\frac{1}{2}\,,\qquad\frac{1}{s_\alpha}+\frac{1}{\:\widehat{s}_\alpha}\;=\;\frac{1}{2}\,,\qquad\alpha\in\{1,2,12\}\,.
 \end{equation}
\end{itemize}

We remark that assumptions (A1)-(A4) above are cast in an ``operational'' form that is immediately ready to be exploited in our proofs, whereas the precise constraints that they impose on the potentials $A_1, A_2, U_1, U_2, V_1, V_2, V_{12}$ are left in a somewhat implicit form -- observe, for instance that a priori conditions (A3) and (A4)  select a sub-class of potentials from condition (A2). It is however easy to recognise that (A1)-(A4) cover a wide range of practically relevant cases (analogously to what observed already in \cite[Section 3.2]{kp-2009-cmp2010}), including for example the inter-particle Coulomb interactions $V_\alpha(x)=c_\alpha |x|^{-1}$, $\alpha\in\{1,2,12\}$ for ordinary one-body Hamiltonians $h_1=h_2=-\Delta$.

In particular, concerning the non-emptiness of assumption (A4), the global-in-time well-posedness of the non-linear Cauchy problem  associated with \eqref{eq:Hartree_system} holds for generic (i.e., not too singular) potentials irrespective of the sign of the interaction: this is due to the fact that the cubic non-linearity is \emph{non-local} (i.e., of convolution form $V*|\varphi|^2$) and hence energy sub-critical, in full analogy to what happens with the usual one-component non-linear Schr\"{o}dinger equation \cite[Corollary 6.1.2]{cazenave}. For local non-linearities as in the system \eqref{eq:mathematical_GP_system} below, which are expected to emerge from the more realistic scaling \eqref{eq:GP_scaling_for_potentials}, finite time blow-up phenomena may instead occur, depending on whether the interactions are attractive or repulsive -- see, e.g., \cite{Ma-Zhao-JMP2008_coupledNLS,Li-Wu-Lai-JPhysA2010-sharp_blowup_thresh_coupledNLS,Jungel-Weishaupl2013_2compNLS_blowup}: in this case one still has existence and uniqueness \emph{locally} in time, but the analog of Theorem \ref{main} below would then make only sense at any fixed time in the interval of local well-posedness of the non-linear problem.

%\cite{Lin-Wei-2006PhysD-coupledNLS,Fanelli-Montefusco_JPhysA_2007_coupledNLS,Ma-Zhao-JMP2008_coupledNLS,Chen-Guo-JMP2009_blowup-profile_2coupledNLS,Li-Wu-Lai-JPhysA2010-sharp_blowup_thresh_coupledNLS,Jungel-Weishaupl2013_2compNLS_blowup}.

We can now state our main result.

\begin{theorem} \label{main}
Consider a two-species bosonic system under assumptions (A1)-(A4) above.
Suppose, at time $t=0$, $\Psi_{N_1,N_2}\in\mathcal{D}[H_{N_1,N_2}]\cap \cH_{N_1,N_2,\mathrm{sym}}$ with $\|\Psi_{N_1,N_2}\|_2=1$, and $(u_0,v_0)\in X$ (the space defined in \eqref{eq:energy_space_1body}) with $\|u_0\|_2=\|v_0\|_2=1$. Correspondingly, for $t\in\mathbb{R}$ let $\Psi_{N_1,N_2}(t):=e^{-\ii t H_{N_1,N_2}}\Psi_{N_1,N_2}$ be the unique solution in $C(\mathbb{R},\mathcal{D}[H_{N_1,N_2}]\cap \cH_{N_1,N_2,\mathrm{sym}})$ to the many-body Schr\"{o}dinger equation
\begin{equation}\label{eq:manybody-Schr}
\ii\partial_t\Psi_{N_1,N_2}(t)\;=\;H_{N_1,N_2}\Psi_{N_1,N_2}(t)\,,\qquad \Psi_{N_1,N_2}(0)=\Psi_{N_1,N_2}\,,
\end{equation}
and let $(u_t,v_t)$ be the unique solution 
in $C(\mathbb{R},X)$ (the space \eqref{soluzioni}-\eqref{eq:energy_space_1body} of assumption (A4))
to the initial value problem consisting of the Hartree system \eqref{eq:Hartree_system} 
with initial condition $(u_0, v_0)$ at $t=0$.
Let  $\gamma^{(1,1)}_{N_1,N_2}(t)$ be the double reduced density matrix associated with $\Psi_{N_1,N_2}(t)$, given by \eqref{eq:def_double_partial_trace-KERNEL}, and define
\begin{equation}
\alpha^{(1,1)}_{N_1,N_2}(t)\;:=\;1-\big\langle u_t\otimes v_t\:,\:\gamma_{N_1,N_2}^{(1,1)}(t)\;u_t\otimes v_t\big\rangle\,.
\end{equation}
Assume further that in the limit $N_1,N_2\to\infty$ the two populations have given asymptotic ratios $c_1,c_2>0$, according to \eqref{eq:N1N2ratio}.
Then there exists a constant $\kappa=\kappa(c_1,c_2)>0$, such that 
\begin{equation}\label{tesi}
\alpha^{(1,1)}_{N_1,N_2}(t)\;\leqslant\;\left(\alpha_{N_1,N_2}^{(1,1)}(0)+\frac{1}{N_1+N_2}\right)e^{\,\kappa f(t)},
\end{equation}
where
\[
\begin{split}
f(t)\;&:=\;\|V_1\|_{L^{r_1}+L^{s_1}}\int_0^t\,\ud\tau\,\left(\|u_\tau\|_{\widehat{r}_1}+\|u_\tau\|_{\widehat{s}_1}\right)\\
&+\|V_2\|_{L^{r_2}+L^{s_2}}\int_0^t\,\ud\tau\,\left(\|v_\tau\|_{\widehat{r}_2}+\|v_\tau\|_{\widehat{s}_2}\right)\\
&+\|V_{12}\|_{L^{r_{12}}+L^{s_{12}}}\int_0^t\,\ud\tau\,\left(\|u_\tau\|_{\widehat{r}_{12}}+\|u_\tau\|_{\widehat{s}_{12}}+\|v_\tau\|_{\widehat{r}_{12}}+\|v_\tau\|_{\widehat{s}_{12}}\right)\,.
\end{split}
\]
\end{theorem}

We are clearly interested in applying Theorem \ref{main} to the case where the initial many-body state of the mixture displays double condensation in the orbitals $u_0$ and $v_0$, in the sense of the discussion of Subsection \ref{subsec:manybd-onebd-condensate-mixture} and of the asymptotics \eqref{eq:def_100BEC-2component} therein. This, together with the estimates of Section \ref{sec:indicators_of_convergence} for the indicators of convergence, leads to the following:

\begin{corollary}\label{cor}
If, in addition to the hypothesis of Theorem \ref{main}, the sequence of initial data $(\Psi_{N_1,N_2})_{N_1,N_2}$ satisfies
\begin{equation}\label{eq:double_BEC_at_tzero}
\alf{1}{1}_{N_1,N_2}(0)\;\leqslant\;K\,\frac{1}{N_1+N_2}\,,
\end{equation}
for some constant $K$ that depends only on the population fractions $c_1$ and $c_2$, 
then $\forall t\in\mathbb{R}$ one has
\begin{equation}\label{eq:persistence_of_BEC_alphasense}
\alpha^{(1,1)}_{N_1,N_2}(t)\;\leqslant\;(K+1)\,\frac{1}{N_1+N_2}\,e^{\,\kappa f(t)}
\end{equation}
and also
\begin{equation}\label{eq:persistence_of_BEC_traceasense}
\mathrm{Tr}\:\big|\,\gamma_{N_1,N_2}^{(1,1)}(t)-|u_t\otimes v_t\rangle\langle u_t\otimes v_t| \,\big|\;\lesssim\;(K+1)\,\frac{1}{\sqrt{N_1+N_2}}\,e^{\,\kappa f(t)/2}\,.
\end{equation}
\end{corollary}

Corollary \ref{cor} provides therefore a quantitative proof of the persistence of the double condensation in the mixture at any finite time. Observe that no matter how faster than $(N_1+N_2)^{-1}$ is the asymptotic BEC \eqref{eq:double_BEC_at_tzero} at $t=0$, the bound \eqref{tesi} always gives at later times a rate of convergence of magnitude  $(N_1+N_2)^{-1}$ in \eqref{eq:persistence_of_BEC_alphasense}. We emphasize also that the exponential deterioration in time of all the above controls \eqref{tesi}, \eqref{eq:persistence_of_BEC_alphasense}, \eqref{eq:persistence_of_BEC_traceasense} of BEC along the time evolution is certainly non optimal and is rather a consequence of the Gr\"{o}nwall-type estimate at the basis of the proof of Theorem \ref{main}.

\subsection{Further consequences and remarks}

\emph{Controls in the energy space.} By replacing assumptions (A2) and (A4) above with
\begin{itemize}
 \item[(A2')] The potentials $V_\alpha$, $\alpha\in\{1,2,12\}$ are real-valued, even, and such that
 \begin{equation}\label{eq:assumptions_on_V_alternative}
 \begin{split}
 \|V_j^2*|\phi_j|^2\|_\infty\;&\lesssim\;\|\phi_j\|^2_{\mathcal{D}[h_j]}\qquad\forall\phi_j\in\mathcal{D}[h_j]\qquad j=1,2 \\
 \|V_{12}^2*|\phi_j|^2\|_\infty\;&\lesssim\;\|\phi_j\|^2_{\mathcal{D}[h_j]}\qquad\forall\phi_j\in\mathcal{D}[h_j]\qquad j=1,2
 \end{split}
 \end{equation}
 \item[(A4')] The initial value problem consisting of the system \eqref{eq:Hartree_system} with initial condition $u(0)=u_0$ and $v(0)=v_0$ for given functions $u_0\in\mathcal{D}[h_1]$ and $v_0\in\mathcal{D}[h_2]$
  has a unique global-in-time solution
 \begin{equation}\label{soluzioni-alternative}
 (u,v)\;\in\;C(\mathbb{R},\mathcal{D}[h_1]\oplus\mathcal{D}[h_2])\cap C^1(\mathbb{R},\mathcal{D}[h_1]^*\oplus\mathcal{D}[h_2]^*)
 \end{equation}
\end{itemize}
then Theorem \ref{main} and Corollary \ref{cor} hold with
\[
\begin{split}
f(t)\;&:=\;\int_0^t\,\ud\tau\,\Big(\|u_\tau\|^2_{\mathcal{D}[h_1]}+\|v_\tau\|^2_{\mathcal{D}[h_2]}\Big)\,.
\end{split}
\]
Indeed, assumptions \eqref{eq:assumptions_on_V_alternative} take the role of Lemma \ref{stime} below and our proof remains virtually unchanged.

\emph{Persistence of condensation for a fraction of the two populations.} It follows straightforwardly from the properties of the
indicators of condensation discussed in Lemma \ref{lemma:intertwined-controls} of Section \ref{sec:indicators_of_convergence} below that the conclusion of Corollary \ref{cor} implies also
\begin{equation}\label{eq:persistence_of_BEC_alphasense-level-k}
\alpha^{(k_1,k_2)}_{N_1,N_2}(t)\;\lesssim\;\frac{\:\max\{k_1,k_2\}\:}{N_1+N_2}\,e^{\,\kappa f(t)}\qquad\forall t\in\mathbb{R}\,,
\end{equation}
which is interpreted as the control of the persistence in time of condensation for $o(N_j)$ particles of the $j$-th species, $j=1,2$.

\emph{More singular potentials.}  Although we are not interested in discussing in full generality the class of interaction potentials that can be dealt with by the present method, it is worth remarking that with a moderate additional effort one can adapt the proof of Theorem \ref{main} (in the spirit of \cite[Section 5]{kp-2009-cmp2010}) so as to include potentials with stronger singularities than those admitted by Assumptions (A1)-(A4) above.

\emph{Control of condensation separately in each component.} In a similar setting to the one analysed here, T.~Heil \cite{Heil-2012}
has discussed the large $N_1,N_2$ asymptotics separately for each one-component reduced density matrix, $\gamma_{N_1,N_2}^{(1,0)}(t)$ and $\gamma_{N_1,N_2}^{(0,1)}(t)$ in our notation. As remarked already with our bound \eqref{eq:bound_g11_g10_g01}, such a control is covered by our collective indicator $\gamma_{N_1,N_2}^{(1,1)}(t)$.

\section{Partial marginals and indicators of condensation}\label{sec:indicators_of_convergence}

We have already introduced in Subsection \ref{subsec:manybd-onebd-condensate-mixture} the notion of double reduced density matrix $\gamma_{N_1,N_2}^{(k_1,k_2)}$ associated with a given state $\gamma_{N_1,N_2}$ for a mixture of $N_1+N_2$ particles and relative to a choice of $k_j\leqslant N_j$ particles of the $j$-th species, $j=1,2$, see equations \eqref{eq:def_double_partial_trace-ABSTRACT}-\eqref{eq:def_double_partial_trace-KERNEL} and the discussion thereafter. In this Section we elaborate further on these indicators and on equivalent quantitative characterisations of asymptotic BEC.

Explicitly, for a pure state (a normalised function) $\Psi_{N_1,N_2}\in\cH_{N_1,N_2,\mathrm{sym}}$ the associated marginal
\begin{equation}\label{eq:def_double_partial_trace-k1-k2}
\gamma_{N_1,N_2}^{(k_1,k_2)}\;=\;\mathrm{Tr}_{N_1-k_1}\otimes\mathrm{Tr}_{N_2-k_2}\:|\Psi_{N_1,N_2}\rangle\langle\Psi_{N_1,N_2}| 
\end{equation}
has integral kernel
\begin{equation}\label{eq:def_double_partial_trace-KERNEL-k1-k2}
\begin{split}
& \gamma_{N_1,N_2}^{(k_1,k_1)}(x_1,\dots,x_{k_1},x_1',\dots, x_{k_1}';y_1,\dots,y_{k_2},y_1',\dots,y_{k_2}')\;=\; \\
&=\int_{\mathbb{R}^{3(N_1-k_1)}} \int_{\mathbb{R}^{3(N_2-k_2)}} \ud x_{k_1+1}\cdots\ud x_{N_1}\ud y_{k_2+1}\cdots\ud y_{N_2} \\
& \;\;\times \Psi_{N_1,N_2}(x_1,\dots,x_{k_1},x_{k_1+1},\dots,x_{N_1};y_1,\dots,y_{k_2},y_{k_2+1},\dots,y_{N_2}) \\
& \;\;\times \overline{\Psi_{N_1,N_2}}(x_1',\dots, x_{k_1}',x_{k_1+1},\dots,x_{N_1};y_1',\dots,y_{k_2}',y_{k_2+1},\dots,y_{N_2})\,.
\end{split}
\end{equation}
As long as $k_1,k_2\geqslant 1$,  $\gamma_{N_1,N_2}^{(k_1,k_2)}$ is a density matrix acting on $\cH_{k_1,k_2,\mathrm{sym}}$. In the extremal cases, \eqref{eq:def_double_partial_trace-k1-k2}-\eqref{eq:def_double_partial_trace-KERNEL-k1-k2} above define reduced density matrices $\gamma_{N_1,N_2}^{(k_1,0)}$ and $\gamma_{N_1,N_2}^{(0,k_2)}$ acting, respectively, on the single-species bosonic spaces $\cH_{k_1}$ and $\cH_{k_2}$, that is, all the degrees of freedom of one of the two components are traced out.

For given one-body orbitals $u,v\in\mathfrak{h}$ with $\|u\|=\|v\|=1$, which are going to play the role of condensate functions for the two component of the mixture, and for given $k_j\leqslant N_j$, $j=1,2$, we define
\begin{equation}\label{eq:indicator-alpha}
\alpha^{(k_1,k_2)}_{N_1,N_2}\;:=\;1-\big\langle u^{\otimes k_1}\otimes v^{\otimes k_1}\,,\,\gamma_{N_1,N_2}^{(k_1,k_2)}\:u^{\otimes k_1}\otimes v^{\otimes k_2}\big\rangle
\end{equation}
and
\begin{equation}\label{eq:indicator-trace}
\ind{R}\;:=\;\tr\,\big|\,\gamma_{N_1,N_2}^{(k_1,k_1)}-|u^{\otimes k_1}\otimes v^{\otimes k_2}\rangle\langle u^{\otimes k_1}\otimes v^{\otimes k_2}|\,\big|\,.
\end{equation}
The scalar product in \eqref{eq:indicator-alpha} and the trace in \eqref{eq:indicator-trace} are taken in the Hilbert space $\cH_{k_1,k_2}$.

Both these indicators measure a displacement of the marginal $\gamma_{N_1,N_2}^{(k_1,k_1)}$ from $|u^{\otimes k_1}\otimes v^{\otimes k_2}\rangle\langle u^{\otimes k_1}\otimes v^{\otimes k_2}|$, which is the $(k_1,k_2)$-reduced density matrix relative to the purely factorised $(N_1+N_2)$-body state $u^{\otimes N_1}\otimes v^{\otimes N_2}$. As already discussed in Subsection \ref{subsec:manybd-onebd-condensate-mixture}, their vanishing for large $N_1,N_2$ has the natural interpretation of occurrence of double (simultaneous) condensation for the two components of the mixture. This is in complete analogy to the indicators of condensation adopted for the one-component case (see equations \eqref{eq:def_100BEC}-\eqref{eq:equivalent-BEC-control-1compontent} and the following comments in the Introduction), which in the present notation are precisely $\alpha^{(k_1,0)}_{N_1,N_2}$ and $R_{N_1,N_2}^{(k_1,0)}$ relative to the first component (and the analogs for the second component).

We discuss the properties of such indicators in these  two  Lemmas and in the following observations. For notational convenience we omit the subscripts $N_1,N_2$.

\begin{lemma}\label{lemma:controllo} For the quantities defined in \eqref{eq:indicator-alpha} for $k_1,k_2\in\{0,1\}$, 
one has
\begin{equation}\label{control}
\alpha^{(1,0)}\;\leqslant\;\alpha^{(1,1)}\,,\qquad\alpha^{(0,1)}\;\leqslant\;\alpha^{(1,1)}
\end{equation}
and
\begin{equation}\label{control-from-above}
\alpha^{(1,1)}\;\leqslant\;\alpha^{(1,0)}+\alpha^{(0,1)}\,.
\end{equation}
\end{lemma}

\begin{lemma}\label{lemma:intertwined-controls} For the quantities defined in \eqref{eq:indicator-alpha}-\eqref{eq:indicator-trace} for $k_j\in\{1,\dots,N_j\}$, $j=1,2$, 
one has
\begin{eqnarray}
 \alpha^{(k_1,k_2)}\; &\leqslant& \;R^{(k_1,k_2)}\;\leqslant\;2\sqrt{\alpha^{(k_1,k_2)} } \label{eq:intertwined-alpha-R} \\ 
 \alpha^{(k_1,k_2)}\; &\leqslant& \;\max\{k_1,k_2\}\cdot\alpha^{(1,1)}\,. \label{eq:tower-alpha-control}
\end{eqnarray}
\end{lemma}

In this work we are mainly concerned with the $\alpha$-indicators, that measure the displacement between the maximal possible value  1 and the expectation of the reduced density matrix on the factorised state. The apparently stronger trace norm displacement (the $R$-indicators) turns out to be equivalent to the former, in view of \eqref{eq:intertwined-alpha-R}. Lemma \ref{lemma:controllo} shows that for the control of the double condensation, simultaneously in each component, one can equivalently monitor the vanishing of $\alpha^{(1,1)}$ or the vanishing of $\alpha^{(1,0)}$ \emph{and} $\alpha^{(0,1)}$. Lemma \ref{lemma:intertwined-controls} shows in addition that the vanishing of $\alpha^{(1,1)}$ or of higher order $\alpha$-indicators are also equivalent, with a deterioration (the factor $\max\{k_1,k_2\}$ in \eqref{eq:tower-alpha-control}) that depends on the size of the subsystem of particles in each component on which one monitors the absence of correlation, that is, the pure factorisation and hence the occurrence of condensation.

\begin{proof}[Proof of Lemma \ref{lemma:controllo}]
If $(v_n)_{n=1}^\infty$ is an orthonormal basis of $\mathfrak{h}=L^2(\mathbb{R}^3)$ such that $v_1=v$, one has
\[
\langle u\otimes v\,,\gamma^{(1,1)}\,u\otimes v\rangle_{\mathfrak{h}^{\otimes 2}}\;\leqslant\;\sum_{n=1}^\infty\langle u\otimes v_n\,,\gamma^{(1,1)}\,u\otimes v_n\rangle_{\mathfrak{h}^{\otimes 2}}\;=\;\langle u\,,\gamma^{(1,0)}u\rangle_{\mathfrak{h}}
\]
where the inequality is due to the positivity of $\gamma^{(1,1)}$ and in the following identity we used the definition of partial trace. This shows that $\alpha^{(1,0)}\leqslant \alpha^{(1,1)}$, and analogously $\alpha^{(0,1)}\leqslant \alpha^{(1,1)}$. To prove \eqref{control-from-above}, we exploit the positivity of the projections $\mathbbm{1}-|u\rangle\langle u|$ and $\mathbbm{1}-|v\rangle\langle v|$, and hence of their tensor product: one finds
\[
\begin{split}
 0\;&\leqslant\;\tr_{\mathfrak{h}^{\otimes 2}}\big[\,\gamma^{(1,1)}(\mathbbm{1}_{\mathfrak{h}}-|u\rangle\langle u|)\otimes(\mathbbm{1}_{\mathfrak{h}}-|v\rangle\langle v|)\,\big] \\
 &=\;\tr_{\mathfrak{h}^{\otimes 2}}\big[\,\,\gamma^{(1,1)}(\mathbbm{1}_{\mathfrak{h}}-|u\rangle\langle u|)\otimes\mathbbm{1}_{\mathfrak{h}}\,\big]
+ \tr_{\mathfrak{h}^{\otimes 2}}\big[\,\,\gamma^{(1,1)} \,\mathbbm{1}_{\mathfrak{h}}\otimes(\mathbbm{1}_{\mathfrak{h}}-|v\rangle\langle v|)\,\big] \\
&\qquad -\tr_{\mathfrak{h}^{\otimes 2}}\big[\, \gamma^{(1,1)} (\mathbbm{1}_{\mathfrak{h}^{\otimes 2}}-|u\rangle\langle u|\otimes|v\rangle\langle v|)  \,\big] \\
& = \; \alpha^{(1,0)}+\alpha^{(0,1)}-\alpha^{(1,1)}
\end{split}
\]
and the conclusion follows.
\end{proof}

\begin{proof}[Proof of Lemma \ref{lemma:intertwined-controls}] Inequalities \eqref{eq:intertwined-alpha-R} are established precisely as in the one-component case, since here one only deals with the rank-one projection $|u^{\otimes k_1}\otimes v^{\otimes k_2}\rangle\langle u^{\otimes k_1}\otimes v^{\otimes k_2}|$ and the density matrix $\gamma^{(k_1,k_2)}$ on the $(k_1+k_2)$-body space $\cH_{k_1,k_2}$, with no reference to the two-component structure or to the numbers $k_1,k_2$ -- see, e.g., \cite[Lemma 2.3]{kp-2009-cmp2010} (the constant $2$ in the second inequality in \eqref{eq:intertwined-alpha-R} is an easy improvement of the constant $2\sqrt{2}$ obtained in \cite[Lemma 2.3]{kp-2009-cmp2010}). As for \eqref{eq:tower-alpha-control}, one first repeats component-wise the very same argument that allows for a control of the $k$-th marginal in terms of the $(k-1)$-th marginal precisely as in the one-component case \cite[Lemma 2.1]{kp-2009-cmp2010}, thus obtaining
\[
\alpha^{(k_1,k_2)}\;\leqslant\;\alpha^{(k_1-1,k_2-1)}+\alpha^{(1,1)}\,.
\]
By iteration, and supposing for example $k_1<k_2$, one gets
\[
\alpha^{(k_1,k_2)}\;\leqslant\;k_1\alpha^{(1,1)}+\alpha^{(0,k_2-k_1)}\,.
\]
In turn, for the one-component $(k_2-k_1)$-marginal $\alpha^{(0,k_2-k_1)}$ the standard one-component argument \cite[Lemma 2.1]{kp-2009-cmp2010} yields
\[
\alpha^{(0,k_2-k_1)}\;\leqslant\;(k_2-k_1)\,\alpha^{(0,1)}
\]
and by Lemma \ref{lemma:controllo} $\alpha^{(0,1)}\leqslant\alpha^{(1,1)}$.  By combining these inequalities,
\[
\alpha^{(k_1,k_2)}\;\leqslant\;k_1\alpha^{(1,1)}+(k_2-k_1)\,\alpha^{(0,1)}\;\leqslant\;k_2\,\alpha^{(1,1)}
\]
which proves \eqref{eq:tower-alpha-control}.
\end{proof}

\section{Mean-field scaling for a two-component mixture}\label{sec:mean-field-scaling}

The purpose of this Section is to justify the mean-field scaling factors chosen in \eqref{eq:HN1N2} for the many-body Hamiltonian $H_{N_1,N_2}$. It was already remarked after \eqref{eq:HN1N2}  that there is large freedom in finding $N_1,N_2$-dependent pre-factors that ensure that the kinetic and the potential part of $H_{N_1,N_2}$ remain of the same order as $N_1,N_2\to\infty$. One then has to make a choice driven by physical considerations, given the asymptotic fractions $c_1$ and $c_2$ of the two populations of particles. %For the sake of simplicity, with no loss of generality throughout this Section we shall assume that $N_j/(N_1+N_2)=c_j$, $j=1,2$, at every finite $N_1,N_2$, and not only asymptotically.

We focus first on the expected system \eqref{eq:Hartree_system} of non-linear Schr\"{o}dinger equations for the effective dynamics of the mixture. 
It is a well familiar fact in  the quest for effective evolution equations of many-body quantum dynamics that performing a mean-field scaling in the Hamiltonian produces at the effective level non-linear Schr\"{o}dinger equations of Hartree type, that is, with ``non-local'' non-linearities of the form $(V*|u|^2)u$ -- see, e.g., \cite[Chapter 2]{Benedikter-Porta-Schlein-2015}. 
Realistic interaction potentials have strong magnitude on a short range and as opposite to the mean-field picture each particle interacts with the neighbouring others on a short scale only: the non-linearity expected in the effective dynamics is then ``local'', namely of the form $g|u|^2 u$. For the two-component condensate mixture under consideration, an amount of physical heuristics and experimental observations \cite[Section 12.11]{ps2003} show that the effective dynamics is governed by the system
\begin{equation}\label{eq:physical_GP_system}
\begin{split}
\ii\partial_t u_t\;&=\;h_1 u_t + g_1 N_1|u_t|^2 u_t + g_{12} N_2 |v_t|^2 u_t \\
\ii\partial_t v_t\;&=\;h_2 v_t + g_2 N_2|v_t|^2 v_t + g_{12} N_1 |u_t|^2 v_t
\end{split}
\end{equation}
as an accurate approximation in the regime of ultra-low temperature, high dilution, and large $N_1,N_2$. In \eqref{eq:physical_GP_system} the couplings $g_1,g_2,g_{12}$ are, in suitable units, the two-body scattering lengths of the corresponding particle-particle interactions and the functions $\sqrt{N_1}\,u$, $\sqrt{N_2}\,v$, with $\|u\|_2=\|v\|_2=1$, represent the so-called ``order parameters'' of the two components of the mixture.

For a mathematical derivation of \eqref{eq:physical_GP_system} one goes through a suitable scaling of the interaction potentials, and hence of the couplings $g_1(N_1,N_2)$, $g_2(N_1,N_2)$, $g_{12}(N_1,N_2)$, as $N_1,N_2\to\infty$, so as to have asymptotically constant couplings
\begin{equation}\label{eq:renormalising_couplings}
\begin{split}
g_1(N_1,N_2) N_1\to \gamma_1\,,&\qquad g_2(N_1,N_2) N_1\to \gamma_2 \\
g_{12}(N_1,N_2)N_2\to D_1\,,&\qquad g_{12}(N_1,N_2)N_1\to D_2
\end{split}
\end{equation}
in a limiting, $N_1,N_2$-independent version of \eqref{eq:physical_GP_system}.
%In fact in the resulting laws $g_1=F_1/N_1$, $g_2=F_2/N_2$ it is not restrictive to set $F_1=F_2=g_\circ$ (this amounts to a re-definition of the interaction potentials for species 1 and 2).
It is then immediate to see that \eqref{eq:renormalising_couplings}, under the constraint $N_1/N_2\to c_1/c_2$, imposes $D_1=\gamma_{12}c_2$ and $D_2=\gamma_{12}c_1$ for some constant $\gamma_{12}>0$ and hence the scalings
\begin{equation}\label{eq:GP_scaling_in_couplings}
g_1\;=\;\frac{\gamma_1}{N_1}\,,\qquad\!\! g_2\;=\;\frac{\gamma_2}{N_2}\,,\qquad\!\! g_{12}\;=\;\gamma_{12}\frac{c_1}{N_1}\;=\;\gamma_{12}\frac{c_2}{N_2}\;=\;\frac{\gamma_{12}}{N_1+N_2}\,.
\end{equation}
Based on  the definition of scattering length for short-range potentials \cite[Appendix C]{LSeSY-ober}, it is easy to see that the re-scaling \eqref{eq:GP_scaling_in_couplings} of the scattering lengths $g_1$, $g_2$, $g_{12}$ is reproduced by scaling the interaction potentials as
\begin{equation}\label{eq:GP_scaling_for_potentials}
\begin{split}
V_1(x)\;&=\;N_1^2 \mathcal{V}_1(N_1x)\,,\qquad V_2(x)\;=\;N_2^2 \mathcal{V}_2(N_2x)\,, \\
V_{12}(x)\;&=\;(N_1+N_2)^2 \mathcal{V}_{12}((N_1+N_2)x)\,,
\end{split}
\end{equation}
for fixed potentials $ \mathcal{V}_1$, $ \mathcal{V}_2$, $ \mathcal{V}_{12}$ with scattering length, respectively, $\gamma_1$, $\gamma_2$, $\gamma_{12}$: this is precisely the Gross-Pitaevskii scaling, namely the case $\beta=1$ in \eqref{eq:scaling_limits}. In this scaling, the non-linear effective system to be derived takes therefore the form
\begin{equation}\label{eq:mathematical_GP_system}
\begin{split}
\ii\partial_t u_t\;&=\;h_1 u_t + \gamma_1 |u_t|^2 u_t + c_2\,\gamma_{12} |v_t|^2 u_t \\
\ii\partial_t v_t\;&=\;h_2 v_t +\gamma_2|v_t|^2 v_t + c_1\,\gamma_{12} |u_t|^2 v_t\,.
\end{split}
\end{equation}

As in the present work we put emphasis on the mechanism for the emergence of a non-linear effective description and for this we content ourselves to discuss the \emph{mean-field} problem, \eqref{eq:mathematical_GP_system} has to be converted to its mean-field form: using that the Born approximation for the scattering length is $\gamma_{\alpha}=\int V_\alpha$, $\alpha\in\{1,2,12\}$, we see that for a formal short-range potential $V_\alpha(x)=\gamma_\alpha\delta(x)$ the counterpart of the non-linearities $\gamma_\alpha |f|^2 g$ has the form $(V_\alpha*|f|^2)g$. The mean-field version of \eqref{eq:mathematical_GP_system}
is therefore the Hartree system \eqref{eq:Hartree_system}. This also indicates that the mean-field version of the Gross-Pitaevskii scaling \eqref{eq:GP_scaling_for_potentials} for the interaction potentials is, in view of \eqref{eq:scaling_limits},
\begin{equation}\label{eq:MF_scaling_for_potentials}
\begin{split}
V_1(x)\;&=\;N_1^{-1} \mathcal{V}_1(x)\,,\qquad V_2(x)\;=\;N_2^{-1} \mathcal{V}_2(x)\,, \\
V_{12}(x)\;&=\;(N_1+N_2)^{-1} \mathcal{V}_{12}(x)\,,
\end{split}
\end{equation}
which accounts for the actual mean-field pre-factors chosen in \eqref{eq:HN1N2}.

We find it instructive to present two alternative arguments for the emergence of the choice \eqref{eq:MF_scaling_for_potentials} in the many-body Hamiltonian. In the first one we allow for generic mean-field pre-factors
\begin{equation}\label{eq:generic_Valpha_MFscaling}
V_\alpha(x)\;=\;m_\alpha(N_1,N_2)\,\mathcal{V}(x)\,,\qquad\alpha\in\{1,2,12\}\,,
\end{equation}
in the expression \eqref{eq:HN1N2} for the many-body Hamiltonian $H_{N_1,N_2}$ and we compute the asymptotics of the energy per particle for the purely double condensate state with orbitals $u$ and $v$, thus finding
\begin{equation}
\begin{split}
 \mathscr{E}[u,v]\;&=\;\lim_{N_1,N_2\to\infty}\frac{\langle u^{\otimes N_1}\otimes v^{\otimes N_2}\,,\,H_{N_1,N_2}\,u^{\otimes N_1}\otimes v^{\otimes N_2}\rangle}{N_1+N_2} \\
 &=\;c_1\langle u,h_1 u\rangle +c_2\langle v,h_2 v\rangle  +\frac{c_1^2 k_1}{2}\,\langle u, V_1*|u|^2 u\rangle  \\
 &\qquad +\frac{c_2^2 k_2}{2}\,\langle v, V_2*|v|^2 v\rangle + c_1 c_2 k_{12} \langle u, V_{12}*|v|^2 u\rangle\,,
\end{split}
\end{equation}
where we set
\begin{equation}\label{eq:def_k}
k_\alpha\;:=\;\lim_{N_1,N_2\to\infty} (N_1+N_2)\cdot m_\alpha\,,\qquad\alpha\in\{1,2,12\}\,.
\end{equation}
The energy $\mathscr{E}$ has to be conserved along the time evolution generated by $H_{N_1,N_2}$: making the Ansatz that the pure factorisation of the initial state is preserved in time, an easy computation shows that for smooth enough solutions $(u_t,v_t)$ to the system
\begin{equation}\label{eq:Hartree_system_generic_MF_scalings}
\begin{split}
\ii\partial_t u_t\;&=\;h_1 u_t + c_1 k_1(V_1*|u_t|^2) u_t + c_2 k_{12} (V_{12}*|v_t|^2) u_t \\
\ii\partial_t v_t\;&=\;h_2 v_t + c_2 k_2(V_2*|v_t|^2) v_t + c_1 k_{12} (V_{12}*|u_t|^2) v_t
\end{split}
\end{equation}
one has
\begin{equation}
\frac{\ud}{\ud t}\mathscr{E}[u_t,v_t]\;=\;0\,,
\end{equation}
which shows that \eqref{eq:Hartree_system_generic_MF_scalings} is the correct effective description of the many-body dynamics under the formal factorisation Ansatz. The comparison between \eqref{eq:Hartree_system_generic_MF_scalings} and the expected \eqref{eq:Hartree_system} gives
\[
c_1 k_1\;=\;1\,,\qquad c_2 k_2 \;=\;1\,,\qquad c_2 k_{12}\;=\;c_2\,,\qquad c_1 k_{12}\;=\;c_1\,,
\]
whence, owing to \eqref{eq:def_k},
\begin{equation}\label{eq:here_are_the_scalings}
m_1\;=\;N_1^{-1}\,,\qquad m_2\;=\;N_2^{-1}\,,\qquad m_{12}\;=\;(N_1+N_2)^{-1}\,.
\end{equation}
With the choice \eqref{eq:here_are_the_scalings}, \eqref{eq:generic_Valpha_MFscaling} reproduces precisely \eqref{eq:MF_scaling_for_potentials}.

Alternatively, one can check the correctness of the chosen mean-field pre-factors by plugging the formal Ansatz of complete factorisation into  the hierarchy of coupled PDE's that must be satisfied by the reduced marginals $\gamma_{N_1,N_2}^{(k_1,k_2)}(t)$ associated with the solution $\Psi_{N_1,N_2}(t)$ to the many-body Schr\"{o}dinger equation $\ii\partial_t\Psi_{N_1,N_2}(t)=H_{N_1,N_2}\Psi_{N_1,N_2}(t)$. This is the so-called BBGKY hierarchy \cite[Chapter 2]{Benedikter-Porta-Schlein-2015}, that consists of coupled equations for the $(k_1,k_2)$-marginals for each $k_1$ and $k_2$ in $\{1,\dots,N_j\}$, the first of which is
\begin{equation}\label{schifo}
\begin{split}
\ii\de_t\gamma^{(1,1)}_{N_1,N_2}(t)
\;=&\;\,\big[\,(h_1)_1^A,\gamma^{(1,1)}_{N_1,N_2}\,\big]+\big[\,(h_2)^B_1,\gamma^{(1,1)}_{N_1,N_2}\,\big]+\\
&+\dfrac{N_1-1}{N_1}\,\text{Tr}_{x_2}\big[\,(V_1(x_1-x_2))^A,\gamma^{(2,1)}_{N_1,N_2}\,\big]+\\
&+\dfrac{N_2-1}{N_2}\,\text{Tr}_{y_2}\big[\,(V_2(y_1-y_2))^B,\gamma^{(1,2)}_{N_1,N_2}\,\big]+\\
&+\dfrac{1}{N_1+N_2}\,\text{Tr}_{x_2,y_2}\big[\,(N_2-1)V_{12}(x_1-y_2) \\
& +(N_1-1)V_{12}(y_1-x_2)+V_{12}(x_1-y_1),\gamma^{(2,2)}_{N_1,N_2}\,\big]\,.
\end{split}
\end{equation}
In \eqref{schifo} and in the following  the notation $\tr_{x_2}$ and its analogs denotes the partial trace over the degrees of freedom of the second particle of the first species, etc. Taking formally $\gamma_{N_1,N_2}^{(k_1,k_2)}(t)\to\gamma_{\infty,\infty}^{(k_1,k_2)}(t)$ as $N_1,N_2\to\infty$ yields the limiting \emph{infinite} BBGKY hierarchy for which \eqref{schifo} takes the limiting form
\begin{equation}\label{inf}
\begin{split}
&\ii\de_t\gamma^{(1,1)}_{\infty,\infty,t}\;=\;\,\big[\,(h_1)^A_1,\gamma^{(1,1)}_{\infty,\infty,t}\,\big]+\big[\,(h_2)^B_1,\gamma^{(1,1)}_{\infty,\infty,t}\,\big]+\\
&+\text{Tr}_{x_2}\big[\,(V_1(x_1-x_2))^A,\gamma^{(2,1)}_{\infty,\infty,t}\,\big]+\text{Tr}_{y_2}\big[\,(V_2(y_1-y_2))^B,\gamma^{(1,2)}_{\infty,\infty,t}\,\big]\\
&+\text{Tr}_{x_2,y_2}\big[\,c_2\,V_{12}(x_1-y_2)+c_1\,V_{12}(y_1-x_2),\gamma^{(2,2)}_{\infty,\infty,t}\,\big] \,.
\end{split}
\end{equation}
With a direct computation one then checks that the formal Ansatz $\gamma^{(k_1,k_2)}_{\infty,\infty}(t)=\left(\ket{u_t}\bra{u_t}\right)^{\otimes k_1}\otimes\left(\ket{v_t}\bra{v_t}\right)^{\otimes k_2}$ in the limit of infinitely many particles, where $(u_t,v_t)$ is a solution to the Hartree system \eqref{eq:Hartree_system}, produces a solution to \eqref{inf} and in fact to all the other equations of the infinite BBGKY hierarchy.

\section{Proof of Theorem \ref{main}}\label{sec:main_proof}

In this Section we present the proof of our main result, Theorem \ref{main}, and of Corollary \ref{cor}. We shall discuss it in various steps.

The proof goes through a suitable modification of Pickl's counting method for the dynamics of a single-component condensate
\cite{kp-2009-cmp2010,Pickl-JSP-2010,Pickl-LMP-2011,Pickl-RMP-2015}, in order to deal with the inter-species interaction terms and the new mean-field coupling factors.
This method is specifically tailored for the quantity $\alpha_{N_1,N_2}^{(1,1)}(t)$ and it is designed to control it in terms of its value at $t=0$. In fact, what in an appropriate sense is actually ``counted'' is, informally speaking, the number of ``bad'' particles of each species in the many-body state $\Psi_{N_1,N_2}(t)$ which are not of the type $u_t$ or $v_t$, more precisely which are described by a one-body orbital orthogonal to $u_t$ or $v_t$. The quantity of interest, according to this interpretation, is therefore the expectation of the single-orbital observable $|u_t\otimes v_t\rangle\langle u_t\otimes v_t|$ on $\mathfrak{h}\otimes\mathfrak{h}$ on the state $\Psi_{N_1,N_2}(t)\in\cH_{N_1,N_2,\mathrm{sym}}$, and hence the quantity
\begin{equation}\label{eq:quantity_to_monitor_alpha}
\begin{split}
\big\langle\,\Psi_{N_1,N_2}(t),(\mathbbm{1}&-|u_t\otimes v_t\rangle\langle u_t\otimes v_t|\,)\,\Psi_{N_1,N_2}(t)\,\big\rangle\;= \\
&=\;1-\big\langle u_t\otimes v_t\:,\:\gamma_{N_1,N_2}^{(1,1)}(t)\;u_t\otimes v_t\big\rangle\;=\;\alpha^{(1,1)}_{N_1,N_2}(t)\,.
\end{split}
\end{equation}

In order to obtain the bound \eqref{tesi} in Theorem \ref{main} we shall establish the following estimate on the time derivative of  $\alpha^{(1,1)}_{N_1,N_2}(t)$:
\begin{equation}\label{eq:differential_Gronwall}
\frac{\ud}{\ud t}\alpha^{(1,1)}_{N_1,N_2}(t)\;\leqslant\; B(t)\,\alpha^{(1,1)}_{N_1,N_2}(t) + \frac{B(t)}{N_1+N_2}\,,\quad t\in\mathbb{R}\,,
\end{equation}
for some function $B(t)$ that is given in terms of certain norms of the interaction potentials $V_1$, $V_2$, $V_{12}$ and of the solutions $u_t$, $v_t$ to the Hartree system, and \emph{independent} of the number of particles.
Explicitly,
\begin{equation}\label{eq:Bt}
\begin{split}
&B(t)=\kappa\Big(\|V_1\|_{L^{r_1}+L^{s_1}}\,\big(\|u\|_{\widehat{r}_1}+\|u\|_{\widehat{s}_1}\big)+\|V_2\|_{L^{r_2}+L^{s_2}}\,\big(\|v\|_{\widehat{r}_2}+\|v\|_{\widehat{s}_2}\big)  \\
&\qquad\qquad +\|V_{12}\|_{L^{r_{12}}+L^{s_{12}}}\big(\|u\|_{\widehat{r}_{12}}+\|u\|_{\widehat{s}_{12}}+\|v\|_{\widehat{r}_{12}}+\|v\|_{\widehat{s}_{12}}\big)\Big)
\end{split}
\end{equation}
for some constant $\kappa$ that depends only on the population fractions $c_1$ and $c_2$. Comparing \eqref{eq:alpha-C-splitting-final} with \eqref{eq:differential_Gronwall}, we obtain \eqref{eq:eq-after-Gronwall} and hence the thesis \eqref{tesi}.

After an integration in time \eqref{eq:differential_Gronwall} gives
\begin{equation}
\alpha^{(1,1)}_{N_1,N_2}(t)\;\leqslant\;\alpha^{(1,1)}_{N_1,N_2}(0)+\frac{1}{N_1+N_2}\int_0^t B(s)\,\ud s+\int_0^t B(s)\,\alpha^{(1,1)}_{N_1,N_2}(s)\,\ud s
\end{equation}
which is of the form
\[
\alpha(t)\;\leqslant\;\beta(t)+\int_0^t\gamma(s)\,\alpha(s)\ud s
\]
for $\beta(t)\equiv\alpha^{(1,1)}_{N_1,N_2}(0)+(N_1+N_2)^{-1}\!\int_0^t B(s)\,\ud s$ and $\gamma(t)\equiv B(t)$
and hence implies the Gr\"{o}nwall-like estimate \cite[Theorem 1.3.2]{Pachpatte-ineq}
\[
\alpha(t)\;\leqslant\;\beta(t)+\int_0^t\beta(s)\,\gamma(s)\,e^{\int_s^t\gamma(r)\,\ud r}\ud s\,.
\]
By further integrations by parts we finally conclude that \eqref{eq:differential_Gronwall} implies
\begin{equation}\label{eq:eq-after-Gronwall}
\alpha^{(1,1)}_{N_1,N_2}(t)\;\leqslant\;\Big(\alpha^{(1,1)}_{N_1,N_2}(0)+\frac{1}{N_1+N_2}\Big)\,e^{\int_0^t B(s)\ud s}\,.
\end{equation}
The bound \eqref{eq:eq-after-Gronwall} above, together with \eqref{eq:Bt}, leads to \eqref{tesi}.

\subsection{Additional notation}\label{subsec:notation_for_the_proof}

An amount of simplified notation will be useful from now on.

We shall drop the $t$-variable and $N$-subscripts, thus setting $\Psi\equiv\Psi_{N_1,N_2}(t)$ for the solution to the many-body Schr\"{o}dinger equation, $u\equiv u_t$ and $v\equiv v_t$ for the solutions to the Hartree system \eqref{eq:Hartree_system}, and $\alpha^{(1,1)}\equiv\alpha^{(1,1)}_{N_1,N_2}(t)$ for the corresponding quantity \eqref{eq:quantity_to_monitor_alpha}. Also, we shall denote the time derivative as $\dot{\alpha}^{(1,1)}$.

We shall keep the convention $T^A$ (resp., $T^B$) for $T\otimes\mathbbm{1}$ (resp., $\mathbbm{1}\otimes T$) where $T$ is an operator  that acts only on one of the two factors of $\cH_{N_1,N_2}$ and we need to consider it as an operator on the whole $\cH_{N_1,N_2}$ with trivial action on the other factor. When, in particular, $T$ is a one-particle operator (that is, $T$ acts on $\mathfrak{h}$), the notation $T^A_j$ for some $j\in\{1,\dots,N_1\}$ or  $T^B_\ell$ for some $\ell\in\{1,\dots,N_2\}$ indicates that we are considering $T$ as acting non-trivially on the one-body space of the $j$-th particle of type $A$ or the $\ell$-th particle of type $B$. In analogy to this convention, we shall write $T_{ij}$ when a \emph{two}-body operator $T$ (i.e., an operator on $\mathfrak{h}\otimes\mathfrak{h}$, meant to be the one-body spaces of each component) acts on $\cH_{N_1,N_2}$ non-trivially only in the variables $x_i$ and $y_j$ of the wave-function $\Psi(x_1,\dots,x_{N_1};y_1,\dots,y_{N_2})$. Observe that this is clearly not to be confused with the symbol $V_{12}$ for the inter-species potential: when needed, according to the convention above we shall rather write $(V_{12})_{ij}$ as a multiplication operator.

Henceforth we shall also omit the explicit tensor product notation $\otimes$: this will leave product expressions that it will be straightforward to interpret as tensor products based on the context.

A special notation for a number of relevant one-particle operators will be convenient. By $h^u$ and $h^v$ we shall denote the two ``one-body non-linear Hamiltonians''
\begin{equation}\label{eq:onebody_nonlin_Hamilt}
\begin{split}
h^u\;&:=\;h_1+V_1*|u_t|^2+c_2 V_{12}*|v_t|^2 \\
h^v\;&:=\;h_2+V_2*|v_t|^2+c_1 V_{12}*|u_t|^2
\end{split}
\end{equation}
and by $p^A$, $p^B$ and $q^A$, $q^B$ we shall denote the orthogonal projections
\begin{equation}
\begin{split}
p^A\;:=\;|u_t\rangle\langle u_t|\,,&\qquad q^A\;:=\;\mathbbm{1}-|u_t\rangle\langle u_t| \\
p^B\;:=\;|v_t\rangle\langle v_t|\,,\,&\qquad q^B\;:=\;\mathbbm{1}-|v_t\rangle\langle v_t|\,.
\end{split}
\end{equation}

Furthermore, we shall make use of the shorthands
\begin{equation}\label{eq:shorthnd_Vu_etc}
\begin{split}
V_1^u\;:=\;V_1*|u_t|^2\,,&\qquad V_2^v\;:=\;V_2*|v_t|^2 \\
V_{12}^u\;:=\;V_{12}*|u_t|^2\,,&\qquad V_{12}^v\;:=\;V_{12}*|v_t|^2\,.
\end{split}
\end{equation}
Observe that according to our convention $(V_1^u)^A_i$ denotes the multiplication operator by the function $(V_1*|u_t|^2)(x_i)$ in the $i$-th of the variables for the species A, and so on.

If $f^\varphi$ is any of the shorthands \eqref{eq:shorthnd_Vu_etc} for some functions $f$ and $\varphi$, then   in terms of the above conventions one has
\begin{equation}\label{eq:trick_f-phi}
p_2^A f_{12}^A\,p_2^A\;=\;p_2^A (f^\varphi)_1^A\;=\;p_2^A (f^\varphi)_1^A p_2^A
\end{equation}
as an identity of two-body operators acting on the A-sector of the many-body Hilbert space -- here $f^A_{12}$ is the function $f(x_1-x_2)$ -- and the same holds for the B-sector. Analogously,
\begin{equation}\label{eq:trick_f-phi-distinct-components}
\begin{split}
p_1^A f_{11}\,p_1^A\;&=\;p_1^A (f^\varphi)_1^B\;=\;p_1^A (f^\varphi)_1^B p_1^A \\
p_1^B f_{11}\,p_1^B\;&=\;p_1^B (f^\varphi)_1^A\;=\;p_1^B (f^\varphi)_1^A p_1^B
\end{split}
\end{equation}
as an identity of mixed-component two-body operators -- here $f_{11}$ is the function $f(x_1-y_1)$.

\subsection{Estimates on convolutions}\label{sec:estimates}

As we will systematically need to bound the $L^\infty$-norm of functions of the form $V*|\phi|^2$ or $V^2*|\phi|^2$, where $V=V_1,V_2,V_{12}$ and $\phi=u_t,v_t$, we cast two standard estimates in the following Lemma.

\begin{lemma} \label{stime}
For given $r,s$ such that $2\leqslant r\leqslant s \leqslant \infty$ let $\widehat{r}$ and $\widehat{s}$ be defined by
\[
\frac{1}{r}+\frac{1}{\widehat{r}}\;=\;\frac{1}{2}\,,\qquad\frac{1}{s}+\frac{1}{\widehat{s}}\;=\;\frac{1}{2}\,.
\]
Then, for $V\in L^r(\mathbb{R}^d)+L^s(\mathbb{R}^d)$ and $\phi\in L^2(\mathbb{R}^d)\cap L^{\widehat{r}}(\mathbb{R}^d)$ with $\|\phi\|_2=1$ one has $\phi\in L^{\widehat{s}}(\mathbb{R}^d)$ and moreover
\begin{equation}\label{convoluno}
\big\|V*|\phi|^2\big\|_\infty\;\leqslant\; \|V\|_{L^{r}+L^{s}}\big(\|\phi\|_{\widehat{r}}+\|\phi\|_{\widehat{s}}\big)
\end{equation}
and
\begin{equation}\label{convoldue}
\big\|V^2*|\phi|^2\big\|_\infty\;\leqslant\; 2\,\|V\|^2_{L^{r}+L^{s}}\big(\|\phi\|_{\widehat{r}}+\|\phi\|_{\widehat{s}}\big)^2\,.
\end{equation}
\end{lemma}

\begin{proof}
By assumption one can split $V=V^{(r)}+V^{(s)}$ with $V^{(r)}\in L^r(\mathbb{R}^d)$ and $V^{(s)}\in L^s(\mathbb{R}^d)$, Then
\[
\begin{split}
\big\|V*|\phi|^2\big\|_\infty \;&\leqslant\; \|V^{(r)}*|\phi|^2\|_\infty+\|V^{(s)}*|\phi|^2\|_\infty\\
&\leqslant\; \|V^{(r)}\|_{r}\|\phi\|^2_{\frac{2r}{r-1}}+\|V^{(s)}\|_{s}\|\phi\|^2_{\frac{2s}{s-1}}\\
&\leqslant\; \big(\,\|V^{(r)}\|_{r}+\|V^{(s)}\|_{s}\,\big)\,\big(\,\|\phi\|^2_{\frac{2r}{r-1}}+\|\phi\|^2_{\frac{2s}{s-1}}\big) \\
& \leqslant\;\big(\,\|V^{(r)}\|_{r}+\|V^{(s)}\|_{s}\,\big)\,\big(\,\|\phi\|_{\widehat{r}}+\|\phi\|_{\widehat{s}}\big)
\end{split}
\]
where the second step follows by Young's inequality and the last step by interpolation on  $\frac{2r}{r-1}\in[2,\widehat{r}]$ and on $\frac{2s}{s-1}\in[2,\widehat{s}]$, using also the fact that $\|\phi\|_2=1$. By taking the infimum over all decompositions of $V$ one obtains
\[
\big\|V*|\phi|^2\big\|_\infty \;\leqslant\; \|V\|_{L^{r}+L^{s}}\,\big(\|\phi\|_{\widehat{r}}+\|\phi\|_{\widehat{s}}\big)
\]
which proves \eqref{convoluno}. Analogously one finds
\[
\begin{split}
\left\|V^2*|\phi|^2\right\|_\infty\;&\leqslant\; 2\,\big\|\big(V^{(r)}\big)^2*|\phi|^2\big\|_\infty + 2\,\big\|\big(V^{(s)}\big)^2*|\phi|^2\big\|_\infty\\
&\leqslant\; 2\,\|V^{(r)}\|^2_{r}\,\|\phi\|^2_{\frac{2r}{r-2}} +  2\,\|V^{(s)}\|^2_{s}\,\|\phi\|^2_{\frac{2s}{s-2}} \\
&\leqslant\; 2\,\big(\|V^{(r)}\|_{r}+\|V^{(s)}\|_{s}\big)^2\left(\|\phi\|_{\widehat{r}}+\|\phi\|_{\widehat{s}}\right)^2,
\end{split}
\]
using again Young's inequality in the second step. By taking the infimum over all decompositions of $V$ one obtains
\begin{equation*}
\big\|V^2*|\phi|^2\big\|_\infty\;\leqslant\; 2\,\|V\|^2_{L^{r}+L^{s}}\big(\|\phi\|_{\widehat{r}}+\|\phi\|_{\widehat{s}}\big)^2\,.
\end{equation*}
which proves \eqref{convoldue}.
\end{proof}

\subsection{Time derivative of $\alpha^{(1,1)}_{N_1,N_2}(t)$ and cancellation of the kinetic terms}

We intend to differentiate in time the quantity $\alpha^{(1,1)}$ written in the form \eqref{eq:quantity_to_monitor_alpha}, that is,
\begin{equation}
\alpha^{(1,1)}\;=\;\langle\Psi,(\mathbbm{1}-p_1^A p_1^B)\Psi\rangle\,.
\end{equation}

When the time derivative hits the $\Psi$'s, this produces a commutator term $[H_{N_1,N_2},p_1^A p_1^B]$ owing to
\eqref{eq:manybody-Schr}, and this term is well defined because assumptions (A2) and (A4) imply $p_1^A p_1^B\Psi\in\mathcal{D}[H_{N_1,N_2}]$. When instead the time derivative hits $p_1^Ap_1^B$, this produces a commutator term $[(h_1^u)^A+(h_1^v)^B,p_1^Ap_1^B]$ owing to \eqref{eq:Hartree_system}, where $h^u$ and $h^v$ are the operators \eqref{eq:onebody_nonlin_Hamilt}. This term too is well defined: indeed, on the one hand Lemma \ref{stime} together with assumptions (A2) and (A4) implies that the multiplicative parts of $h^u$ and $h^v$  (i.e., the functions $V_1*|u|^2$, $V_2*|v|^2$, $V_{12}*|u|^2$, and $V_{12}*|v|^2$) are all bounded, which in turn implies the boundedness of $h^u$ and $h^v$  as operators $h^{u}:\mathcal{D}[h_1]\to\mathcal{D}[h_1]^*$, $h^{v}:\mathcal{D}[h_2]\to\mathcal{D}[h_2]^*$; on the other hand $p_1^A p_1^B\Psi\in\mathcal{D}[H_{N_1,N_2}]$ as already observed, and $\mathcal{D}[H_{N_1,N_2}]\subset\mathcal{D}[H^{(0)}_{N_1,N_2}]$ owing to assumptions (A3), which makes the expectation $\langle\Psi,[(h_1^u)^A+(h_1^v)^B,p_1^Ap_1^B]\Psi\rangle$ well defined. The conclusion is therefore that $\alpha^{(1,1)}$ is differentiable in time and
\begin{equation}\label{eq:first_derivative_of_alpha}
\dot{\alpha}^{(1,1)}\;=\;\ii\,\langle\Psi,[H_{N_1,N_2}-(h_1^u)^A-(h_1^v)^B,\mathbbm{1}-p_1^A p_1^B]\,\Psi\rangle\,.
\end{equation}

In the r.h.s.~of \eqref{eq:first_derivative_of_alpha} the insertion of further terms $(h_j^u)^A$ and $(h_j^v)^B$ with $j\geqslant 2$ does not produce any effect, since their commutator with $\mathbbm{1}-p_1^A\,p_1^B$ vanishes. This gives
\begin{equation}\label{eq:first_derivative_of_alpha-II}
\dot{\alpha}^{(1,1)}\;=\;\ii\,\langle\Psi,[H_{N_1,N_2}-(H^u)^A-(H^v)^B,\mathbbm{1}-p_1^A p_1^B]\,\Psi\rangle\,.
\end{equation}
where
\begin{equation}
H^u\;:=\;\sum_{k=1}^{N_1}h^u_k\,,\qquad H^v:=\sum_{\ell=1}^{N_2}h^v_\ell\,.
\end{equation}
Further, one can re-write the r.h.s.~of \eqref{eq:first_derivative_of_alpha-II} as the expectation of 
an operator that is completely symmetric in each component, namely
\begin{equation}\label{eq:first_derivative_of_alpha-III}
\dot{\alpha}^{(1,1)}\;=\;\ii\,\Big\langle\Psi\,,\Big[H_{N_1,N_2}-(H^u)^A-(H^v)^B\,,\,\sum_{k=1}^{N_1}\sum_{\ell=1}^{N_2}\dfrac{\mathbbm{1}-p_k^A\,p_\ell^B}{N_1N_2}\,\Big]\,\Psi\Big\rangle\,.
\end{equation}
Observe that when passing from \eqref{eq:first_derivative_of_alpha} to \eqref{eq:first_derivative_of_alpha-II} one obtains a complete \emph{cancellation of the kinetic terms} and they will play no role henceforth. Thus,
\begin{equation}
\begin{split}
\dot{\alpha}^{(1,1)}\;&=\;\ii\,\Big\langle\Psi\,,\Big[\,\dfrac{1}{N_1}\sum_{i<j}^{N_1}(V_1(x_i-x_j))^A+\dfrac{1}{N_2}\sum_{r<s}^{N_2}(V_2(y_r-y_s)^B\\
&\qquad +\dfrac{1}{N_1+N_2}\sum_{i=1}^{N_1}\sum_{r=1}^{N_2}V_{12}(x_i-y_r) \\
&\qquad -\sum_{i=1}^{N_1}(V_1^u)_i^A -c_2\sum_{i=1}^{N_1}(V_{12}^v)_i^A \\
& \qquad -\sum_{r=1}^{N_2}(V_2^v)_r^B-c_1\sum_{r=1}^{N_2}(V_{12}^u)_r^B \,,\,\sum_{k=1}^{N_1}\sum_{\ell=1}^{N_2}\dfrac{\mathbbm{1}-p_k^A\,p_\ell^B}{N_1N_2}\,\Big]\,\Psi\Big\rangle\
\end{split}
\end{equation}
where we used the shorthands  \eqref{eq:shorthnd_Vu_etc}.

We separate the contributions  given to $\dot{\alpha}^{(1,1)}$ by each potential $V_1$, $V_2$, $V_{12}$ and write
\begin{equation}\label{eq:alpha-C-splitting}
\dot{\alpha}^{(1,1)}\;=\;\ii\,(C_{V_1}+C_{V_2}+C_{V_{12}})
\end{equation}
with
\begin{equation}\label{eq:C-V1}
C_{V_1}:=\Big\langle\Psi,\Big[\Big(\dfrac{1}{N_1}\sum_{i<j}^{N_1}V_1(x_i-x_j)-\sum_{i=1}^{N_1}(V_1^u)_i\Big)^A ,\sum_{k=1}^{N_1}\sum_{\ell=1}^{N_2}\dfrac{\mathbbm{1}-p_k^A\,p_\ell^B}{N_1N_2}\Big]\Psi\Big\rangle,
\end{equation}
\begin{equation}\label{eq:C-V2}
C_{V_2}:=\Big\langle\Psi,\Big[\Big(\dfrac{1}{N_2}\sum_{r<s}^{N_2}V_2(y_r-y_s)-\sum_{r=1}^{N_2}(V_2^v)_r\Big)^B ,\sum_{k=1}^{N_1}\sum_{\ell=1}^{N_2}\dfrac{\mathbbm{1}-p_k^A\,p_\ell^B}{N_1N_2}\Big]\Psi\Big\rangle,
\end{equation}
\begin{equation}\label{eq:C-V12}
\begin{split}
C_{V_{12}}&=\Big\langle\Psi,\Big[\dfrac{1}{N_1+N_2}\sum_{i=1}^{N_1}\sum_{r=1}^{N_2}V_{12}(x_i-y_r)-c_2\sum_{i=1}^{N_1}(V_{12}^v)_i^A \\
&\qquad\qquad -c_1\sum_{r=1}^{N_2}(V_{12}^u)_r^B ,\sum_{k=1}^{N_1}\sum_{\ell=1}^{N_2}\dfrac{\mathbbm{1}-p_k^A\,p_\ell^B}{N_1N_2}\Big]\Psi\Big\rangle\,.
\end{split}
\end{equation}

In the following Subsections we shall estimate separately these three terms, see Propositions \ref{prop:CV1} and \ref{prop:CV12} below. The final result, obtained by plugging \eqref{eq:CV1-control} and \eqref{eq:CV12-control} into \eqref{eq:alpha-C-splitting}, is
\begin{equation}\label{eq:alpha-C-splitting-final}
\begin{split}
&\dot{\alpha}^{(1,1)}\;\leqslant\;\kappa\,\Big(\alpha^{(1,1)}+\frac{1}{N_1+N_2}\Big)\,\times \\
&\quad \times\Big(\|V_1\|_{L^{r_1}+L^{s_1}}\,\big(\|u\|_{\widehat{r}_1}+\|u\|_{\widehat{s}_1}\big)+\|V_2\|_{L^{r_2}+L^{s_2}}\,\big(\|v\|_{\widehat{r}_2}+\|v\|_{\widehat{s}_2}\big)  \\
&\qquad\qquad +\|V_{12}\|_{L^{r_{12}}+L^{s_{12}}}\big(\|u\|_{\widehat{r}_{12}}+\|u\|_{\widehat{s}_{12}}+\|v\|_{\widehat{r}_{12}}+\|v\|_{\widehat{s}_{12}}\big)\Big)
\end{split}
\end{equation}
for some constant $\kappa$ that depends only on the population fractions $c_1$ and $c_2$. Comparing \eqref{eq:alpha-C-splitting-final} with \eqref{eq:differential_Gronwall}, we obtain \eqref{eq:eq-after-Gronwall} and hence the thesis \eqref{tesi}.

\subsection{Terms containing $V_1$ and $V_2$}\label{subsec:V1_and_V2_terms}

By means of straightforward commutation properties we re-write \eqref{eq:C-V1} as
\begin{equation}\label{eq:CV1-redone}
\begin{split}
C_{V_1}\;&=\;\Big\langle\Psi,\Big[\Big(\dfrac{1}{N_1}\sum_{i<j}^{N_1}V_1(x_i-x_j)-\sum_{i=1}^{N_1}(V_1^u)_i\Big)^A ,\sum_{k=1}^{N_1}\sum_{\ell=1}^{N_2}\dfrac{-p_k^A\,p_\ell^B}{N_1N_2}\Big]\Psi\Big\rangle \\
&=\;\Big\langle\Psi,\Big[\dfrac{1}{N_1}\sum_{i<j}^{N_1}V_1(x_i-x_j)-\sum_{i=1}^{N_1}(V_1^u)_i,\sum_{k=1}^{N_1}\frac{-p_k}{N_1}\Big]^A p_1^B\,\Psi\Big\rangle \\
&=\;\Big\langle\Psi,\Big[\dfrac{1}{N_1}\sum_{i<j}^{N_1}V_1(x_i-x_j)-\sum_{i=1}^{N_1}(V_1^u)_i\,,\,\widehat{m}\,\Big]^A p_1^B\,\Psi\Big\rangle \\
&=\;\frac{1}{2}\big\langle\Psi,\big[(N_1-1)(V_1)_{12}-N_1(V_1^u)_1-N_1(V_1^u)_2\,,\,\widehat{m}\,\big]^A p_1^B\,\Psi\big\rangle
\end{split}
\end{equation}
where $\widehat{m}$ is the auxiliary operator defined in \eqref{eq:def_f-hat} and \eqref{eq:def_mhat}, and where in the third step we applied property  \eqref{eq:property-m-q} for $\widehat{m}$ and in the last step we exploited the symmetry of $\Psi$. Analogously, from  \eqref{eq:C-V2},
\begin{equation}
\begin{split}
C_{V_2}\;&=\;\frac{1}{2}\big\langle\Psi,\big[(N_2-1)(V_2)_{12}-N_2(V_2^v)_1-N_2(V_2^v)_2\,,\,\widehat{m}\,\big]^B p_1^A\,\Psi\big\rangle\,.
\end{split}
\end{equation}

\begin{proposition}\label{prop:CV1} Under the hypotheses of Theorem \ref{main},
\begin{equation}\label{eq:CV1-control}
\begin{split}
|C_{V_1}|\;&\leqslant\;\kappa_1\Big(\alpha^{(1,1)}+\frac{1}{N_1+N_2}\Big)\,\|V_1\|_{L^{r_1}+L^{s_1}}\,\big(\|u\|_{\widehat{r}_1}+\|u\|_{\widehat{s}_1}\big) \\
|C_{V_2}|\;&\leqslant\;\kappa_2\Big(\alpha^{(1,1)}+\frac{1}{N_1+N_2}\Big)\,\|V_2\|_{L^{r_2}+L^{s_2}}\,\big(\|v\|_{\widehat{r}_2}+\|v\|_{\widehat{s}_2}\big) \,.
\end{split}
\end{equation}
For both $j=1,2$ the constant $\kappa_j$ depends only on the population fraction $c_j$.
\end{proposition}

\begin{proof} We shall focus on $C_{V_1}$, the proof for $C_{V_2}$ is completely analogous. In fact, since the commutator in the r.h.s.~of \eqref{eq:CV1-redone} is non-trivial on the first component only, the treatment of $C_{V_1}$ is analogous to the single-component case. By inserting on both sides of the commutator in \eqref{eq:CV1-redone} the identity
\begin{equation}\label{eq:add_idendity-A}
\mathbbm{1}^A\;=\;(p_1^A+q_1^A)(p_1^A+q_2^A)
\end{equation}
one obtains 16 terms; however, owing to Lemma \ref{lemma:exchange-hat}, only those terms with different numbers of $q$'s on the left and on the right are non-zero (see the remark after \eqref{eq:hat-business-for-2}). We cast them in the following self-explanatory notation
\begin{equation}\label{eq:CV1-symbolic}
C_{V_1}\;=\;2\,(pp,qp)+2\,(qp,qq)+(pp,qq)+\text{complex conjugate}
\end{equation}
We shall estimate each summand above in terms of $\alpha^{(1,1)}$ and $(N_1+N_2)^{-1}$.

The first term is
\[
\begin{split}
(pp,qp)\;&=\;\frac{\ii}{2}\big\langle\Psi,p_1^Ap_2^A\left[(N_1-1)(V_1)_{12}-N_1\,(V_1^u)_1,\widehat{m}\,\right]^Aq_1^Ap_2^Ap_1^B\Psi\big\rangle\\
&=\;\frac{\ii}{2}\big\langle\Psi,p_1^Ap_2^A\left[(N_1-1)(V_{1}^u)_1-N_1\,(V_1^u)_1,\widehat{m}\,\right]^Aq_1^Ap_2^Ap_1^B\Psi\big\rangle\\
&=\;-\frac{\ii}{2}\big\langle\Psi,p_1^Ap_2^A\left[\,(V_1^u)_1,\widehat{m}\,\right]^Aq_1^Ap_2^Ap_1^B\Psi\big\rangle\\
&=-\frac{\ii}{2N_1}\big\langle\Psi,p_1^Ap_2^A (V_1^u)_1^A q_1^Ap_2^Ap_1^B\Psi\big\rangle\,.
\end{split}
\]
where we used $p_1^Aq_1^A=0$ in the first and last identities and property \eqref{eq:trick_f-phi} in the second identity. Therefore, by Lemma \ref{stime},
\begin{equation}\label{eq:pp-qp-C1}
\begin{split}
|(pp,qp)|\;&\leqslant\;\frac{1}{2N_1}\|V_{1}*|u|^2\|_\infty \;\leqslant\; \frac{1}{2N_1}\|V_1\|_{L^{r_1}+L^{s_1}}\,\big(\|u\|_{\widehat{r}_1}+\|u\|_{\widehat{s}_1}\big) \\
&\lesssim\; \frac{1}{c_1}\,\frac{1}{N_1+N_2} \,\|V_1\|_{L^{r_1}+L^{s_1}}\,\big(\|u\|_{\widehat{r}_1}+\|u\|_{\widehat{s}_1}\big) \,.
\end{split}
\end{equation}

Following analogous steps, the second summand in \eqref{eq:CV1-symbolic} becomes
\[
\begin{split}
(qp,qq)\;&=\;\frac{\ii}{2}\big\langle\Psi,q_1^Ap_2^A[(N_1-1)(V_{1})_{12}-N_1(V_1^u)_{2},\widehat{m}\,]^Aq_1^Aq_2^Ap_1^B\Psi\big\rangle\\
&=\;\frac{\ii}{2}\Big\langle\Psi,q_1^Ap_2^A\Big(\frac{N_1-1}{N_1}(V_1)_{12}-(V_1^u)_2\Big)^Aq_1^Aq_2^Ap_1^B\Psi\Big\rangle\,.
\end{split}
\]
Splitting the difference we obtain two terms: the second is controlled by a Cauchy-Schwarz inequality and by estimate \eqref{convoluno} of Lemma \ref{stime}  as
\[
\begin{split}
\dfrac{1}{2}|\langle\Psi,q_1^Ap_2^A(V_1^u)_2^A q_1^Aq_2^Ap_1^B\Psi\rangle|\;&\leqslant\;\frac{1}{2}\,\|V_1*|u|^2\|_\infty\,\|q_1^A\psi\|^2\\
&\leqslant\; \frac{1}{2}\,\|V_1\|_{L^{r_1}+L^{s_1}}\,\big(\|u\|_{\widehat{r}_1}+\|u\|_{\widehat{s}_1}\big)\,\alf{1}{1},
\end{split}
\]
having bounded $\|q_1^A\Psi\|^2=\alpha^{(1,0)}$ with $\alpha^{(1,1)}$ (Lemma \ref{lemma:controllo}). The first term is again controlled by Cauchy-Schwarz as 
\[
\begin{split}
\dfrac{1}{2}|\langle\Psi,q_1^Ap_2^A &(V_1)_{12}^A\, q_1^Aq_2^Ap_1^B\Psi\rangle| \\
&\leqslant\; \frac{1}{2}\sqrt{\langle\Psi,q_1^Ap_2^A((V_1)_{12}^2)^A\,p_2^Aq_1^A\Psi\rangle}\,\sqrt{\langle\Psi,q_1^Aq_2^Ap_1^B\Psi\rangle}\\
&=\;\frac{1}{2}\sqrt{\langle\Psi,q_1^Ap_2^A (V_1^2*|u|^2)_1^A p_2^Aq_1^A\Psi\rangle}\,\sqrt{\langle\Psi,q_1^Aq_2^Ap_1^B\Psi\rangle}\\
&\leqslant\; \frac{1}{2}\sqrt{\|V_1^2*|u|^2\|_\infty\,}\;\|q_1^A\psi\|\,\|q_2^A\psi\|\\
&=\;\frac{1}{2}\,\sqrt{\left\|V_1^2*|u|^2\right\|_\infty}\,\alpha^{(1,1)}
\end{split}
\]
having used Lemma \ref{lemma:controllo} in the last step. Then, owing to  estimate \eqref{convoldue} of Lemma \ref{stime}, 
\[
 \dfrac{1}{2}|\langle\Psi,q_1^Ap_2^A (V_1)_{12}^A\, q_1^Aq_2^Ap_1^B\Psi\rangle|\;\leqslant\;\frac{1}{\sqrt{2}}\,\|V_1\|_{L^{r_1}+L^{s_1}}\,\big(\|u\|_{\widehat{r}_1}+\|u\|_{\widehat{s}_1}\big)\,\alf{1}{1}
\]
and the conclusion is
\begin{equation}\label{eq:qp-qq-C1}
|(qp,qq)|\;\lesssim\; \|V_1\|_{L^{r_1}+L^{s_1}}\,\big(\|u\|_{\widehat{r}_1}+\|u\|_{\widehat{s}_1}\big)\,\alf{1}{1}\,.
\end{equation}

The third summand in \eqref{eq:CV1-symbolic} reads
\[
\begin{split}
(pp,qq)\;&=\;\frac{\ii}{2}\big\langle\Psi,p_1^Ap_2^A[(N_1-1)(V_1)_{12},\widehat{m}\,]^Aq_1^Aq_2^Ap_1^B\Psi\big\rangle\\
&=\;\ii\,\frac{N_1-1}{N_1}\big\langle\Psi,p_1^Ap_2^A(V_1)_{12}^A\,\widehat{n}^A(\widehat{ n}^{-1})^Aq_1^Aq_2^Ap_1^B\Psi\big\rangle\\
&=\;\ii\,\frac{N_1-1}{N_1}\big\langle\Psi,p_1^Ap_2^A\,\widehat{\tau_2 n}^A(V_1)_{12}^A(\widehat{ n}^{-1})^Aq_1^Aq_2^Ap_1^B\Psi\big\rangle\,,
\end{split}
\]
where in the second step we applied Lemma \eqref{lemma:exchange-hat} and 
we introduced the auxiliary operator $\widehat{n}$ defined in  \eqref{eq:def_f-hat} and \eqref{eq:def_m_n}, using the fact that $(\widehat{n}^{-1})^A$ is well defined on the range of $q_1^A$ since $(\widehat{n}^{-1})^Aq_1^A\Psi=\sum_{k=1}^N(N/k)^{1/2}P_kq_1\Psi$, while the last identity follows from Lemma \ref{lemma:exchange-hat}  in the form \eqref{eq:hat-business-for-2}.
Then
\[
\begin{split}
|(pp&,qq)|\;\leqslant\;\sqrt{\big\langle\Psi,p_1^Ap_2^A\widehat{\tau_2 n}^A((V_1)_{12}^2)^A\,\widehat{\tau_2 n}^Ap_1^Ap_2^A\Psi\big\rangle}\sqrt{\big\langle\Psi,(\widehat{n}^{-2})^Aq_1^Aq_2^Ap_1^B\Psi\big\rangle}\\
&\leqslant\;\sqrt{\big\langle\Psi,p_1^Ap_2^A\,\widehat{\tau_2 n}^A (V_1^2*|u|^2)_1^A\,\widehat{\tau_2 n}^Ap_1^Ap_2^A\Psi\big\rangle}\,\sqrt{\frac{N_1}{N_1-1}}\,\|q_2^A\,\Psi\|\\
&\leqslant\;2\,\sqrt{\|V_1^2*|u|^2\|_\infty\,}\,\big\|\widehat{\tau_2 n}^A\psi\big\|\,\sqrt{\alpha^{(1,1)}\,}\\
&=\;2\,\sqrt{\|V_1^2*|u|^2\|_\infty\,}\,\sqrt{\alpha^{(1,1)}+\frac{2}{N_1}\,}\,\sqrt{\alf{1}{1}}\\
&\leqslant\;\frac{4}{c_1}\,\sqrt{\|V_1^2*|u|^2\|_\infty\,}\,\Big(\alpha^{(1,1)}+\frac{1}{N_1+N_2}\Big)
\end{split}
\]
where we used the Cauchy-Schwarz inequality in the first step, 
\eqref{eq:trick_f-phi} and \eqref{eq:commutation-modified-application1} of Lemma \ref{commutazione-partial-symmetry}  in the second,
the control $\alpha^{(1,0)}\leqslant\alpha^{(1,1)}$ (Lemma \ref{lemma:controllo}) in the third,
\eqref{eq:def_f-hat}, \eqref{eq:def_m_n}, and \eqref{eq:shifted_function} in the fourth, 
and 
\[
\begin{split}
\sqrt{\alpha^{(1,1)} +\frac{2}{N_1}}\,\sqrt{\alpha^{(1,1)}}\;&\leqslant\;\alpha^{(1,1)}+\frac{1}{N_1}\;\leqslant\;\frac{N_1+N_2}{N_1}\,\Big(\alpha^{(1,1)}+\frac{1}{N_1+N_2}\Big) \\
&\leqslant\; \frac{2}{c_1}\,\Big(\alpha^{(1,1)}+\frac{1}{N_1+N_2}\Big)
\end{split}
\]
in the last, for $N_1,N_2$ sufficiently large.
Then, owing to estimate \eqref{convoldue} of Lemma \ref{stime} we conclude
\begin{equation}\label{eq:pp-qq-C1}
\begin{split}
|(pp&,qq)|\;\leqslant\;\frac{4\sqrt{2}}{c_1}\,\|V_1\|_{L^{r_1}+L^{s_1}}\,\big(\|u\|_{\widehat{r}_1}+\|u\|_{\widehat{s}_1}\big) \left(\alf{1}{1}+\frac{1}{N_1+N_2}\right).
\end{split}
\end{equation}

Plugging \eqref{eq:pp-qp-C1}, \eqref{eq:qp-qq-C1}, and \eqref{eq:pp-qq-C1} into \eqref{eq:CV1-symbolic} yields finally \eqref{eq:CV1-control}.
\end{proof}

\subsection{Term containing $V_{12}$}

We first exploit in  \eqref{eq:C-V12}  the asymptotics  $N_j(N_1+N_2)^{-1}\sim c_j$, $j=1,2$ and the symmetry of $\Psi$:
\begin{equation}\label{eq:CV12-simplified}
\begin{split}
|C_{V_{12}}|\;&\leqslant\;\frac{N_1N_2}{N_1+N_2}\,\times \\
&\quad\times\Big|\Big\langle\Psi,\Big[(V_{12})_{11}-(V_{12}^v)_1^A-(V_{12}^u)_1^B ,\sum_{k=1}^{N_1}\sum_{\ell=1}^{N_2}\dfrac{p_k^A\,p_\ell^B}{N_1N_2}\Big]\Psi\Big\rangle\Big|\,. %\\
%\kappa_{c_1c_2}\;&\lesssim\;\frac{N_1N_2}{N_1+N_2}\,.
\end{split}
\end{equation}
For the estimate of $C_{V_{12}}$ we shall establish the following:

\begin{proposition}\label{prop:CV12} Under the hypotheses of Theorem \ref{main},
\begin{equation}\label{eq:CV12-control}
\begin{split}
|C_{V_{12}}|\;&\leqslant\;\kappa_{12}\,\|V_{12}\|_{L^{r_{12}}+L^{s_{12}}}\big(\|u\|_{\widehat{r}_{12}}+\|u\|_{\widehat{s}_{12}}+\|v\|_{\widehat{r}_{12}}+\|v\|_{\widehat{s}_{12}}\big)\,\times \\
&\qquad\qquad\times\Big(\alpha^{(1,1)}+\frac{1}{N_1+N_2}\Big)
\end{split}
\end{equation}
where the constant $\kappa_{12}$ depends only on the population fractions $c_1$ and $c_2$.
\end{proposition}

\begin{proof}
We insert on both sides of the commutator in \eqref{eq:CV12-simplified} the identity
\begin{equation}\label{eq:add_idendity-AB}
\mathbbm{1}\;=\;(p_1^A+q_1^A)(p_1^B+q_1^B)
\end{equation}
(observe that, as opposite to \eqref{eq:add_idendity-A}, in \eqref{eq:add_idendity-AB} the insertion involves \emph{both} components), which produces 16 terms to estimate. Unlike our previous bookkeeping \eqref{eq:CV1-symbolic}, it is not possible to apply Lemma \ref{lemma:exchange-hat} in order to identify a priori those that vanish, because here the  $p$'s and $q$'s inserted on the left and on the right are relative to \emph{distinct} components. We rather group these terms depending on whether the number of the $q$'s is the same or not on both sides, that is, 
\begin{equation}\label{eq:Lambda}
\begin{split}
\Lambda\;&:=\;(pp,pp)+[(pq,pq)+(qp,qp)]+(qq,qq) \\
&\qquad +\left[(pq,qp)+\text{complex conjugate}\,\right]
\end{split}
\end{equation}
and
\begin{equation}\label{eq:Omega}
\begin{split}
\Omega\;&:=\;(pp,qp)+(qp,qq)+(pp,qq)+(pp,pq)+(pq,qq) \\
&\qquad +\text{ complex conjugate }
\end{split}
\end{equation}
where a self-explanatory notation analogous to \eqref{eq:CV1-symbolic} is used.
In Subsections \ref{subsec:same_n_q} and \ref{subsec:different_n_q} below we shall find
\begin{equation}
|\Lambda|\;\lesssim\;\|V_{12}\|_{L^{r_{12}}+L^{s_{12}}}\big(\|u\|_{\widehat{r}_{12}}+\|u\|_{\widehat{s}_{12}}+\|v\|_{\widehat{r}_{12}}+\|v\|_{\widehat{s}_{12}}\big)\:\alpha^{(1,1)}
\end{equation}
(see \eqref{eq:Lambda-final} below) and
\begin{equation}
\begin{split}
|\Omega|\;&\leqslant\;\widetilde{\kappa}_{12}\,\|V_{12}\|_{L^{r_{12}}+L^{s_{12}}}\big(\|u\|_{\widehat{r}_{12}}+\|u\|_{\widehat{s}_{12}}+\|v\|_{\widehat{r}_{12}}+\|v\|_{\widehat{s}_{12}}\big)\,\times \\
&\qquad\qquad\times\Big(\alpha^{(1,1)}+\frac{1}{N_1+N_2}\Big)\,.
\end{split}
\end{equation}
(see \eqref{eq:Omega-final} below), for some constant $\widetilde{\kappa}_{12}$ that depends on $c_1$ and $c_2$ only, which completes the proof.
\end{proof}

\subsubsection{Terms with the same number of $q$'s on the left and on the right}\label{subsec:same_n_q}

In order to apply a number of straightforward symmetry and permutation arguments it will be convenient to re-write systematically
\begin{equation}\label{eq:sum_expanded}
\Big[A_{11}\,,\,\sum_{k=1}^{N_1}\sum_{\ell=1}^{N_2}p_k^A\,p_\ell^B\Big]\;=\;\Big[A_{11}\,,\,\sum_{k=1}^{N_1}p_k^A\,p_1^B+\sum_{\ell=1}^{N_2}p_1^A\,p_\ell^B-p_1^A p_1^B\Big]
\end{equation}
whenever we deal with an observable $A_{11}$ acting on the first variable of each component.

The summand $(pp,pp)$ in \eqref{eq:Lambda} vanishes because
\[
\begin{split}
&%\frac{|(pp,pp)|}{\kappa_{c_1c_2}/(N_1N_2)}\;\leqslant\;
\Big\langle\Psi,p_1^A p_1^B\Big[(V_{12})_{11}-(V_{12}^v)_1^A-(V_{12}^u)_1^B ,\sum_{k=1}^{N_1}\sum_{\ell=1}^{N_2}p_k^A\,p_\ell^B\Big]p_1^A p_1^B\Psi\Big\rangle \\
&=\Big\langle\Psi,p_1^A p_1^B\Big[(V_{12})_{11}\!-\!(V_{12}^v)_1^A\!-\!(V_{12}^u)_1^B\,, \sum_{k=1}^{N_1}p_k^A\,p_1^B+\sum_{\ell=1}^{N_2}p_1^A\,p_\ell^B-p_1^A p_1^B\Big]p_1^Ap_1^B\Psi\Big\rangle \\
&=\;0
\end{split}
\]
where in the first identity we used \eqref{eq:sum_expanded} and in the second one we used the fact that the $p_1$-operators inside the commutator can be re-absorbed in the corresponding $p_1$'s outside, thus yielding a vanishing commutator.

The summand $(pq,pq)$ in \eqref{eq:Lambda} vanishes because
\[
\begin{split}
&%\frac{|(pp,pp)|}{\kappa_{c_1c_2}/(N_1N_2)}\;\leqslant\;
\Big\langle\Psi,p_1^A q_1^B\Big[(V_{12})_{11}-(V_{12}^v)_1^A-(V_{12}^u)_1^B ,\sum_{k=1}^{N_1}\sum_{\ell=1}^{N_2}p_k^A\,p_\ell^B\Big]p_1^A q_1^B\Psi\Big\rangle \\
&=\Big\langle\Psi,p_1^A q_1^B\Big[(V_{12})_{11}\!-\!(V_{12}^v)_1^A\!-\!(V_{12}^u)_1^B\,, \sum_{k=1}^{N_1}p_k^A\,p_1^B+\sum_{\ell=1}^{N_2}p_1^A\,p_\ell^B-p_1^A p_1^B\Big]p_1^Aq_1^B\Psi\Big\rangle \\
&=\Big\langle\Psi,p_1^A q_1^B\Big[(V_{12})_{11}\!-\!(V_{12}^v)_1^A\!-\!(V_{12}^u)_1^B\,, \sum_{\ell=2}^{N_2}p_\ell^B\Big]q_1^B\Psi\Big\rangle \;=\;0
\end{split}
\]
where in the first identity we used \eqref{eq:sum_expanded}, in the second identity we used $p_1^B q_1^B=\mathbb{O}$ and we re-absorbed $p_1^A$ outside the commutator, and in the last one we used the fact that the two entries of the commutator act on different variables. Obviously, the summand  $(qp,qp)$ in \eqref{eq:Lambda} vanishes for the same reason, upon exchanging the roles of $A$ and $B$.

The summand $(qq,qq)$ in \eqref{eq:Lambda} vanishes owing to $pq=\mathbb{O}$, indeed
\[
\begin{split}
&%\frac{|(pp,pp)|}{\kappa_{c_1c_2}/(N_1N_2)}\;\leqslant\;
\Big\langle\Psi,q_1^A q_1^B\Big[(V_{12})_{11}-(V_{12}^v)_1^A-(V_{12}^u)_1^B \,,\,\sum_{k=1}^{N_1}\sum_{\ell=1}^{N_2}p_k^A\,p_\ell^B\Big]q_1^A q_1^B\Psi\Big\rangle\;=\;0\,,
\end{split}
\]

Thus, in order to estimate the quantity $\Lambda$ in \eqref{eq:Lambda} it only remains to give a bound to the term of type $(pq,qp)$. One has
\[
\begin{split}
&%\frac{|(pp,pp)|}{\kappa_{c_1c_2}/(N_1N_2)}\;\leqslant\;
\frac{N_1N_2}{N_1+N_2}\Big\langle\Psi,p_1^A q_1^B\Big[(V_{12})_{11}-(V_{12}^v)_1^A-(V_{12}^u)_1^B ,\sum_{k=1}^{N_1}\sum_{\ell=1}^{N_2}\frac{p_k^A\,p_\ell^B}{N_1N_2}\Big]q_1^A p_1^B\Psi\Big\rangle \\
&\!\!=\;\frac{1}{N_1+N_2}\,\Big\langle\Psi,p_1^A q_1^B\Big[(V_{12})_{11}\,,\,\sum_{k=1}^{N_1}\sum_{\ell=1}^{N_2}p_k^A\,p_\ell^B\Big]q_1^A p_1^B\Psi\Big\rangle \\
&\!\!=\;\frac{1}{N_1+N_2}\,\Big\langle\Psi,p_1^A q_1^B\Big[(V_{12})_{11}\,,\,\sum_{k=1}^{N_1}p_k^A\,p_1^B+\sum_{\ell=1}^{N_2}p_1^A\,p_\ell^B\Big]q_1^A p_1^B\Psi\Big\rangle \\
&\!\!=\;\frac{N_1-1}{N_1+N_2}\big\langle\Psi,p_1^A q_1^B(V_{12})_{11}p_2^Aq_1^A p_1^B\Psi\big\rangle-\frac{N_2-1}{N_1+N_2}\big\langle\Psi,p_1^A q_1^B p_2^B(V_{12})_{11}q_1^A p_1^B\Psi\big\rangle
\end{split}
\]
where in the first identity the summand $(V_{12}^v)_1^A$ (respectively $(V_{12}^u)_1^B $) does not contribute because $p_1^B$ (resp., $q_1^A$) from the right can be pulled through the commutator all the way to the left with $q_1^Bp_1^B=\mathbb{O}$ (resp., $p_1^A q_1^A=\mathbb{O}$), in the second identity we used \eqref{eq:sum_expanded} and again $pq=\mathbb{O}$, and in the last identity we used the fact that one term of each commutator vanishes because of $pq=\mathbb{O}$.

Therefore, since asymptotically $(N_j-1)(N_1+N_2)^{-1}\leqslant c_j\leqslant 1$, $j=1,2$,
\begin{equation}\label{eq:further-pqqp}
\begin{split}
|\Lambda|\;&\leqslant\;\big|\big\langle\Psi,p_1^A q_1^B(V_{12})_{11}p_2^Aq_1^A p_1^B\Psi\big\rangle\big|+\big|\big\langle\Psi,p_1^A q_1^B p_2^B(V_{12})_{11}q_1^A p_1^B\Psi\big\rangle\big|\,.
\end{split}
\end{equation}
For the first summand in the r.h.s.~of \eqref{eq:further-pqqp} one has
\[
\begin{split}
 \big|\big\langle\Psi,p_1^A q_1^B&(V_{12})_{11}p_2^Aq_1^A p_1^B\Psi\big\rangle\big|\;\leqslant\;\|(V_{12})_{11}p_1^A q_1^B\Psi\|\,\|q_1^Ap_2^Ap_1^B\Psi\| \\
 &\leqslant\;\sqrt{\langle\Psi,p_1^A q_1^B (V^2_{12})_{11}p_1^A q_1^B\Psi\rangle}\,\sqrt{\langle\Psi,q_1^A\,\Psi\rangle} \\
 &=\;\sqrt{\langle\Psi,p_1^A q_1^B (V^2_{12}*|u|^2)_1^Bp_1^A q_1^B\Psi\rangle}\,\|q_1^A\Psi\| \\
 &\leqslant\;\sqrt{\|V^2_{12}*|u|^2\|_\infty\,}\,\|q_1^B\Psi\| \,\|q_1^A\Psi\|\\
 &\leqslant\; \sqrt{2}\,\|V_{12}\|_{L^{r_{12}}+L^{s_{12}}}\big(\|u\|_{\widehat{r}_{12}}+\|u\|_{\widehat{s}_{12}}\big)\:\alpha^{(1,1)}
\end{split}
\]
where we used the Cauchy-Schwarz inequality in the first step, the operator bound $\mathbb{O}\leqslant p\leqslant\mathbbm{1}$ in the second and fourth step,  identity \eqref{eq:trick_f-phi-distinct-components} in the third step, and the identities $\alpha^{(1,0)}=\|q_1^A\Psi\|^2$ and $\alpha^{(0,1)}=\|q_1^B\Psi\|^2$ in the fourth step together with the bounds \eqref{control}  of Lemma \ref{lemma:controllo} that produce $\alpha^{(1,1)}$.
Along the same line, the second summand in the r.h.s.~of \eqref{eq:further-pqqp} is estimated as
\[
\big|\big\langle\Psi,p_1^A q_1^B p_2^B(V_{12})_{11}q_1^A p_1^B\Psi\big\rangle\big|\;\leqslant\;\sqrt{2}\,\|V_{12}\|_{L^{r_{12}}+L^{s_{12}}}\big(\|v\|_{\widehat{r}_{12}}+\|v\|_{\widehat{s}_{12}}\big)\:\alpha^{(1,1)}\,.
\]
The conclusion is 
\begin{equation}\label{eq:Lambda-final}
|\Lambda|\;\lesssim\;\|V_{12}\|_{L^{r_{12}}+L^{s_{12}}}\big(\|u\|_{\widehat{r}_{12}}+\|u\|_{\widehat{s}_{12}}\|v\|_{\widehat{r}_{12}}+\|v\|_{\widehat{s}_{12}}\big)\:\alpha^{(1,1)}\,.
\end{equation}

\subsubsection{Terms with a different number of $q$'s on the left and on the right}\label{subsec:different_n_q}

We first check that the terms $(pp,qp)$ and $(pp,pq)$ in \eqref{eq:Omega} are zero. Indeed,
\[
\begin{split}
\Big\langle\Psi,\,&p_1^A p_1^B\Big[(V_{12})_{11}-(V_{12}^v)_1^A-(V_{12}^u)_1^B \,,\,\sum_{k=1}^{N_1}\sum_{\ell=1}^{N_2} p_k^A\,p_\ell^B\,\Big]q_1^Ap_1^B\Psi\Big\rangle \\
&\;=\;\Big\langle\Psi,p_1^A p_1^B\Big[(V_{12})_{11}-(V_{12}^v)_1^A \,,\,\sum_{k=1}^{N_1}\sum_{\ell=1}^{N_2} p_k^A\,p_\ell^B\,\Big]q_1^Ap_1^B\Psi\Big\rangle \\
&\;=\;\Big\langle\Psi,p_1^A p_1^B\Big[((V_{12}^v)_1^A-(V_{12}^v)_1^A \,,\,\sum_{k=1}^{N_1}\sum_{\ell=1}^{N_2} p_k^A\,p_\ell^B\,\Big]q_1^Ap_1^B\Psi\Big\rangle\;=\;0\,,
\end{split}
\]
where in the first identity the term $(V_{12}^u)_1^B$ does not contribute because $q_1^A$ from the right can be pulled through the commutator all the way to the left with $p_1^Aq_1^A=\mathbb{O}$, and
in the second identity we applied \eqref{eq:trick_f-phi-distinct-components}.
This shows that $(pp,qp)=0$ and an analogous argument shows that $(pp,pq)=0$.

For the term $(qp,qq)$ in \eqref{eq:Omega} one has
\[
\begin{split}
&\dfrac{N_1N_2}{N_1+N_2}\,\Big\langle\Psi,q_1^A p_1^B\Big[(V_{12})_{11}-(V_{12}^v)_1^A-(V_{12}^u)_1^B \,,\,\sum_{k=1}^{N_1}\sum_{\ell=1}^{N_2} \,\frac{p_k^A\,p_\ell^B}{N_1N_2}\,\Big]q_1^Aq_1^B\Psi\Big\rangle\\
&=\;\dfrac{1}{N_1+N_2}\,\Big\langle\Psi,q_1^A p_1^B\Big[(V_{12})_{11}-(V_{12}^u)_1^B \,,\,\sum_{k=1}^{N_1}p_k^A\,p_1^B+\sum_{\ell=1}^{N_2}p_1^A\,p_\ell^B-p_1^A p_1^B\,\Big]q_1^Aq_1^B\Psi\Big\rangle\\
&=\;\dfrac{1}{N_1+N_2}\,\Big\langle\Psi,q_1^A p_1^B\Big[(V_{12})_{11}-(V_{12}^u)_1^B \,,\,\sum_{k=2}^{N_1}p_k^A\,p_1^B\,\Big]q_1^Aq_1^B\Psi\Big\rangle \\
&=\;-\dfrac{(N_1-1)}{N_1+N_2}\,\Big\langle\Psi,q_1^A p_1^B\big((V_{12})_{11}-(V_{12}^u)_1^B \big)\,q_1^Aq_1^Bp_2^A\,\Psi\Big\rangle\,,
\end{split}\]
where in the first identity we applied \eqref{eq:sum_expanded} and we dropped the $(V_{12}^v)_1^A$-term owing to the commutation of $q_1^B$ from the right all the way through to the left with $p_1^Bq_1^B=\mathbb{O}$, in the second identity  we used $p_1^Aq_1^A=\mathbb{O}$, and in the third we used the symmetry of $\Psi$ and again $p_1^Bq_1^B=\mathbb{O}$. In the above quantity, the summand with $(V_{12})_{11}$ can be estimated with the very same arguments used for the control of the first summand in the r.h.s.~of \eqref{eq:further-pqqp} above, that is,
\[
\begin{split}
 \big|\big\langle\Psi,q_1^A p_1^B&(V_{12})_{11}q_1^Aq_1^Bp_2^A\,\Psi\big\rangle\big|\;\leqslant\;\|(V_{12})_{11}q_1^A p_1^B\Psi\|\,\|q_1^A q_1^B p_2^A\Psi\|  \\
 &\leqslant\;\sqrt{\langle\Psi,q_1^A p_1^B (V^2_{12})_{11}q_1^A p_1^B\,\Psi\rangle}\,\sqrt{\langle\Psi,q_1^Aq_1^B\,\Psi\rangle} \\
 &\leqslant\;\sqrt{\langle\Psi,q_1^A p_1^B (V^2_{12}*|v|^2)_1^Aq_1^A p_1^B\,\Psi\rangle}\,\sqrt{\|q_1^A\Psi\|\,\|q_1^B\Psi\|\,} \\
 &\leqslant\;\sqrt{\|V^2_{12}*|v|^2\|_\infty\,}\;\alpha^{(1,1)}\\
 &\leqslant\; \sqrt{2}\,\|V_{12}\|_{L^{r_{12}}+L^{s_{12}}}\big(\|v\|_{\widehat{r}_{12}}+\|v\|_{\widehat{s}_{12}}\big)\:\alpha^{(1,1)}\,.
\end{split}
\]
The summand with $(V_{12}^u)_1^B$ is estimated via a Cauchy-Schwarz inequality and the bound \eqref{convoldue} of Lemma \ref{stime} as
\[
\begin{split}
|\langle\Psi,q_1^A p_1^B\,&(V_{12}^u)_1^B\,q_1^Aq_1^Bp_2^A\,\Psi\rangle|\;\leqslant\;\|V_{12}*|u|^2\|_\infty\,\|q_1^A\Psi\| \,\|q_1^B\Psi\| \\
&\leqslant\;\sqrt{2}\,
\|V_{12}\|_{L^{r_{12}}+L^{s_{12}}}\big(\|u\|_{\widehat{r}_{12}}+\|u\|_{\widehat{s}_{12}}\big)\,\alpha^{(1,1)}\,.
\end{split}
\]
Therefore, since asymptotically $(N_1-1)(N_1+N_2)^{-1}\leqslant c_1\leqslant 1$, 
\begin{equation*}
\begin{split}
|(qp,qq)|\;&\lesssim\;\|V_{12}\|_{L^{r_{12}}+L^{s_{12}}}\big(\|u\|_{\widehat{r}_{12}}+\|u\|_{\widehat{s}_{12}}+\|v\|_{\widehat{r}_{12}}+\|v\|_{\widehat{s}_{12}}\big)\,\alpha^{(1,1)}\,.
\end{split}
\end{equation*}

The very same discussion above for $(qp,qq)$ can be repeated for the term $(pq,qq)$ in \eqref{eq:Omega}. Thus,
\begin{equation}\label{eq:further-qpqq}
\begin{split}
|(qp,qq)|&+|(pq,qq)|\;\lesssim\;\|V_{12}\|_{L^{r_{12}}+L^{s_{12}}}\times \\
&\qquad\times\big(\|u\|_{\widehat{r}_{12}}+\|u\|_{\widehat{s}_{12}}+\|v\|_{\widehat{r}_{12}}+\|v\|_{\widehat{s}_{12}}\big)\,\alpha^{(1,1)}\,.
\end{split}
\end{equation}

It remains to control the term $(pp,qq)$ in \eqref{eq:Omega}. One has
\[
\begin{split}
&\dfrac{N_1N_2}{N_1+N_2}\,\Big\langle\Psi,p_1^A p_1^B\Big[(V_{12})_{11}-(V_{12}^v)_1^A-(V_{12}^u)_1^B \,,\,\sum_{k=1}^{N_1}\sum_{\ell=1}^{N_2} \,\frac{p_k^A\,p_\ell^B}{N_1N_2}\,\Big]q_1^Aq_1^B\Psi\Big\rangle\\
%&=\;\dfrac{1}{N_1+N_2}\,\Big\langle\Psi,p_1^A p_1^B\Big[(V_{12})_{11} \,,\,\sum_{k=1}^{N_1}\sum_{\ell=1}^{N_2} \,p_k^A\,p_\ell^B\,\Big]q_1^Aq_1^B\Psi\Big\rangle\\
&=\;\dfrac{1}{N_1+N_2}\,\Big\langle\Psi,p_1^A p_1^B\Big[(V_{12})_{11} \,,\,\sum_{k=1}^{N_1}p_k^A\,p_1^B+\sum_{\ell=1}^{N_2}p_1^A\,p_\ell^B-p_1^A p_1^B\,\Big]q_1^Aq_1^B\Psi\Big\rangle\\
&=\;\dfrac{N_1-1}{N_1+N_2}\,\langle\Psi,p_1^A p_1^B\big[(V_{12})_{11}\,,\,p_2^Ap_1^B\,\big]q_1^Aq_1^B\Psi\rangle \\
&\qquad\quad +\dfrac{N_2-1}{N_1+N_2}\,\langle\Psi,p_1^A p_1^B\big[(V_{12})_{11}\,,\,p_1^Ap_2^B\,\big]q_1^Aq_1^B\Psi\rangle \\
&\qquad\quad +\dfrac{1}{N_1+N_2}\,\langle\Psi,p_1^A p_1^B\big[(V_{12})_{11}\,,\,p_1^Ap_1^B\,\big]q_1^Aq_1^B\Psi\rangle \\
&\equiv\;(pp,qq)_1+(pp,qq)_2+(pp,qq)_3\,,
\end{split}\]
where in the first identity we applied \eqref{eq:sum_expanded} and we dropped the $(V_{12}^v)_1^A$-term (respectively, the $(V_{12}^u)_1^B$-term) owing to the commutation of $q_1^B$ (resp., $q_1^A$) from the right all the way through to the left with $p_1q_1=\mathbb{O}$, and in the second identity we used the symmetry of $\Psi$.

One has
\[
\begin{split}
(pp,qq)_1\;&=\;\dfrac{N_1-1}{N_1+N_2}\,\langle\Psi,p_1^A p_1^B\big[(V_{12})_{11}\,,\,p_2^Ap_1^B\,\big]q_1^Aq_1^B\Psi\rangle \\
&=\;-\dfrac{N_1-1}{N_1+N_2}\,\langle\Psi,p_1^A p_1^B\,(V_{12})_{11}\,p_2^Aq_1^Aq_1^B\Psi\rangle \\
&=\;-\dfrac{N_1-1}{N_1+N_2}\,\langle\Psi,p_1^A p_1^B\,(V_{12})_{11}\,\widehat{n}^A(\widehat{n}^{-1})^A p_2^Aq_1^Aq_1^B\Psi\rangle \\
&=\;-\dfrac{N_1-1}{N_1+N_2}\,\langle\Psi,p_1^A p_1^B\,\widehat{\tau_1 n}^A(V_{12})_{11}(\widehat{n}^{-1})^A p_2^Aq_1^Aq_1^B\Psi\rangle
\end{split}
\]
where in the third step we introduced the auxiliary operator $\widehat{n}$ defined in  \eqref{eq:def_f-hat} and \eqref{eq:def_m_n}, using again the fact that $(\widehat{n}^{-1})^A$ is well defined on the range of $q_1^A$, and in the last step we applied Lemma \ref{lemma:exchange-hat}  in the form \eqref{eq:hat-business-for-2}. Therefore,
\[
\begin{split}
&|(pp,qq)_1|\;\leqslant\;|\langle\Psi,p_1^A p_1^B\,\widehat{\tau_1 n}^A(V_{12})_{11}(\widehat{n}^{-1})^A p_2^Aq_1^Aq_1^B\,\Psi\rangle| \\
&\qquad \leqslant \;\sqrt{\langle\Psi,p_1^A p_1^B\,\widehat{\tau_1 n}^A(V^2_{12})_{11}\,\widehat{\tau_1 n}^Ap_1^A p_1^B\,\Psi\rangle}\,\sqrt{\langle\Psi,(\widehat{n}^{-2})^A p_2^Aq_1^Aq_1^B\,\Psi\rangle} \\
&\qquad \leqslant \;2\,\sqrt{\langle\Psi,p_1^A p_1^B\,\widehat{\tau_1 n}^A(V^2_{12}*|v|^2)_{1}^B\,\widehat{\tau_1 n}^Ap_1^A p_1^B\,\Psi\rangle}\,\|q_1^B\,\Psi\| \\
&\qquad\leqslant\;2\,\sqrt{\|V^2_{12}*|v|^2\|_\infty}\,\sqrt{\langle\Psi,\widehat{\tau_1 m}^A\Psi\rangle}\,\sqrt{\alpha^{(1,1)}} \\
&\qquad=\;2\,\sqrt{\|V^2_{12}*|v|^2\|_\infty}\,\sqrt{\Big\langle\Psi,\Big(\widehat{m}^A+\frac{1}{N_1}\Big)\Psi\Big\rangle}\,\sqrt{\alpha^{(1,1)}} \\
&\qquad=\;2\,\sqrt{\|V^2_{12}*|v|^2\|_\infty}\,\sqrt{\alpha^{(1,1)} +\frac{1}{N_1}}\,\sqrt{\alpha^{(1,1)}} \\
&\qquad\leqslant\;\frac{4\sqrt{2}}{c_1}\,\|V_{12}\|_{L^{r_{12}}+L^{s_{12}}}\big(\|v\|_{\widehat{r}_{12}}+\|v\|_{\widehat{s}_{12}}\big)\,\Big(\alpha^{(1,1)}+\frac{1}{N_1+N_2}\Big)
\end{split}
\]
where in the first step we used the asymptotic bound $(N_1-1)(N_1+N_2)^{-1}\leqslant c_1\leqslant 1$, in the second we applied the Cauchy-Schwarz inequality, in the third we used \eqref{eq:trick_f-phi-distinct-components} and \eqref{eq:def_m_n}, as well as  Lemma \ref{commutazione-partial-symmetry}
in the form \eqref{eq:commutation-modified-application}, in the fourth we used again \eqref{eq:def_m_n} and $\|q_1^A\Psi\|^2=\alpha^{(1,0)}\leqslant\alpha^{(1,1)}$ (\eqref{control}  of Lemma \ref{lemma:controllo}), in the fifth we used \eqref{eq:def_f-hat}, in the sixth we used \eqref{eq:def_mhat} and again $\|q_1^A\Psi\|^2\leqslant\alpha^{(1,1)}$, and in the last we applied \eqref{convoldue} of Lemma \ref{stime} and
\[
\begin{split}
\sqrt{\alpha^{(1,1)} +\frac{1}{N_1}}\,\sqrt{\alpha^{(1,1)}}\;&\leqslant\;\alpha^{(1,1)}+\frac{1}{2 N_1}\;\leqslant\;\frac{N_1+N_2}{N_1}\,\Big(\alpha^{(1,1)}+\frac{1}{N_1+N_2}\Big) \\
&\leqslant\; \frac{2}{c_1}\,\Big(\alpha^{(1,1)}+\frac{1}{N_1+N_2}\Big)\,.
\end{split}
\]

With the very same arguments one finds
\[
|(pp,qq)_2|\;\leqslant\;\frac{2\sqrt{2}}{c_2}\,\|V_{12}\|_{L^{r_{12}}+L^{s_{12}}}\big(\|u\|_{\widehat{r}_{12}}+\|u\|_{\widehat{s}_{12}}\big)\,\Big(\alpha^{(1,1)}+\frac{1}{N_1+N_2}\Big)\,.
\]

Last,
\[
\begin{split}
|(pp,qq)_3|\;&=\;\Big|\dfrac{1}{N_1+N_2}\,\Big\langle\Psi,p_1^A p_1^B\big[(V_{12})_{11}\,,\,p_1^Ap_1^B\,\big]q_1^Aq_1^B\Psi\Big\rangle\Big| \\
&\leqslant\;\dfrac{1}{N_1+N_2}\,|\langle\Psi,p_1^A p_1^B\,(V_{12})_{11}\,q_1^Aq_1^B\Psi\rangle| \\
&\leqslant\;\dfrac{1}{N_1+N_2}\,\sqrt{\langle\Psi,p_1^A p_1^B\,(V^2_{12})_{11}\,p_1^Ap_1^B\Psi\rangle} \\
&=\;\dfrac{1}{N_1+N_2}\,\sqrt{\langle\Psi,p_1^A p_1^B\,(V^2_{12}*|u|^2)_{1}^A\,p_1^Ap_1^B\Psi\rangle} \\
&\leqslant\;\dfrac{1}{N_1+N_2}\,\sqrt{\|V^2_{12}*|u|^2\|_\infty} \\
&\leqslant\;\dfrac{\sqrt{2}}{N_1+N_2}\,\|V_{12}\|_{L^{r_{12}}+L^{s_{12}}}\big(\|u\|_{\widehat{r}_{12}}+\|u\|_{\widehat{s}_{12}}\big)\,
\end{split}
\]
where the second step follows by  $pq=\mathbb{O}$, the third by
a Cauchy-Schwarz inequality and the operator bound $\mathbb{O}\leqslant q\leqslant\mathbbm{1}$, the fourth by \eqref{eq:trick_f-phi-distinct-components}, the fifth by $\mathbb{O}\leqslant p\leqslant\mathbbm{1}$, and the last by  \eqref{convoldue}.

Therefore,
\begin{equation}\label{eq:further-ppqq}
\begin{split}
|(pp,qq)|\;&\leqslant\;|(pp,qq)_1|+|(pp,qq)_2|+|(pp,qq)_3| \\
&\leqslant\;\widetilde{\kappa}_{12}\,\|V_{12}\|_{L^{r_{12}}+L^{s_{12}}}\big(\|u\|_{\widehat{r}_{12}}+\|u\|_{\widehat{s}_{12}}+\|v\|_{\widehat{r}_{12}}+\|v\|_{\widehat{s}_{12}}\big)\,\times \\
&\qquad\qquad\times\Big(\alpha^{(1,1)}+\frac{1}{N_1+N_2}\Big)
\end{split}
\end{equation}
where the constant $\widetilde{\kappa}_{12}$ only depends on the population fractions $c_1$ and $c_2$.

Plugging \eqref{eq:further-qpqq} and \eqref{eq:further-ppqq} into \eqref{eq:Omega}, which are the only non-zero contributions to $\Omega$, and renaming $\widetilde{\kappa}_{12}$, we finally obtain
\begin{equation}\label{eq:Omega-final}
\begin{split}
|\Omega|\;&\leqslant\;\widetilde{\kappa}_{12}\,\|V_{12}\|_{L^{r_{12}}+L^{s_{12}}}\big(\|u\|_{\widehat{r}_{12}}+\|u\|_{\widehat{s}_{12}}+\|v\|_{\widehat{r}_{12}}+\|v\|_{\widehat{s}_{12}}\big)\,\times \\
&\qquad\qquad\times\Big(\alpha^{(1,1)}+\frac{1}{N_1+N_2}\Big)\,.
\end{split}
\end{equation}

\appendix

\section{Tools exported from the treatment of the single-component case}\label{sec:particle-counting_method}

We collect in this Appendix a number of tools, needed in the proof of Theorem \ref{main}, on which the ``counting'' method 
 is based, quoting their properties from the previous treatments of the single-component case \cite{kp-2009-cmp2010,Pickl-JSP-2010,Pickl-LMP-2011,Pickl-RMP-2015}.

The notation is that introduced in Subsection \ref{subsec:notation_for_the_proof} as applicable for one component only. Thus, in particular, we consider the projections
\begin{equation}
p\;:=\;|\phi\rangle\langle\phi|\,,\qquad q\;:=\;\mathbbm{1}-|\phi\rangle\langle\phi|
\end{equation}
on the one-body Hilbert space $\mathfrak{h}$ and their realisation $p_j$, $q_j$, $j\in\{1,\dots,N\}$ as orthogonal projections on the many-body Hilbert space $\cH_N=\mathfrak{h}^{\otimes N}$, where $\phi\in\mathfrak{h}$ with $\|\phi\|=1$. Clearly, $p+q=\mathbbm{1}$  and $pq=\mathbb{O}=[p,q]$.

Associated to $p$ and $q$ one defines the family of orthogonal projections $P_k$  acting on $\cH_N$, defined by
\begin{equation}
\begin{split}
P_k\;&:=\sum_{\substack{a\in\{0,1\}^N\\
                     \sum_{i}a_i=k}}\;\prod_{i=1}^N \;p_i^{1-a_i}q_i^{a_i}\qquad\textrm{ if } k\in\{0,1,\dots,N\} \\
P_k\;&:=\;\mathbb{O}\qquad\qquad\qquad\qquad\qquad\;\; \textrm{otherwise}\,.
\end{split}
\end{equation}
Each $P_k$ consists by construction of the sum of all possible $N$-fold tensor products of the $p$'s and the $q$'s with $k$ factor $q$. It  therefore arises as the $k$-th term in the expansion of the identity
\begin{equation}
\mathbbm{1}\;=\;(p_1+q_1)\cdots(p_N+q_N)\;=\;\sum_{k=0}^N P_k
\end{equation}
in powers of $q$. It is also clear by the commutation properties of the $p_j$'s and $q_j$'s that
\begin{equation}\label{eq:commutation_of_Pk}
P_k\,P_\ell\;=\;\delta_{k,\ell} P_k\,.
\end{equation}

A relevant role is played by suitable weighted linear combinations of the $P_k$'s. To this aim one introduces, associated to any 
function $f:\{0,1,\dots,N\}\to\mathbb{C}$, i.e., any $(N+1)$-ple $(f(0),\dots,f(N))\in\mathbb{C}^{N+1}$,  the operator
\begin{equation}\label{eq:def_f-hat}  
\widehat{f}\;:=\;\sum_{k=0}^N\, f(k)\, P_k\,.
\end{equation}
As an immediate consequence of \eqref{eq:commutation_of_Pk} and  of the commutation properties of the $p_j$'s,
\begin{equation}
[\,\widehat{f}\,,\,p_j\,]\;=\;[\,\widehat{f}\,,\,P_k\,]\;=\;\mathbb{O}\,,    \qquad[\,\widehat{f}\,,\,\widehat{g}\,]\;=\;\mathbb{O}\,.
\end{equation}
Two convenient choices for the function $f$ shall be
\begin{equation}\label{eq:def_m_n}
 m(k)\;:=\;\frac{k}{N}\,,\qquad n(k)\;:=\;\sqrt{\frac{k}{N}}\,.
\end{equation}
For the operator $\widehat{m}$ one has
\begin{equation}\label{eq:def_mhat}
\frac{1}{N}\sum_{j=1}^N q_j\;=\;\frac{1}{N}\sum_{k=0}^N\sum_{j=1}^N  q_j  P_k\;=\;\frac{1}{N}\sum_{k=0}^N k P_k\;=\;\widehat{m}\,.
\end{equation}
Therefore, if $\Psi\in\cH_{N,\mathrm{sym}}\equiv(\mathfrak{h}^{\otimes N})_{\mathrm{sym}}$, then \eqref{eq:def_mhat} implies
\begin{equation}\label{eq:property-m-q}
\langle\Psi,q_1\,\Psi\rangle\;=\;\langle\Psi,\widehat{m}\,\Psi\rangle\,.
\end{equation}

We thus come to the following useful bounds (see \cite[Lemma 3.9]{kp-2009-cmp2010}):
\begin{lemma}\label{commutazione}
For any $f:\{0,\dots,N\}\rightarrow[0,+\infty)$ and any $\Psi\in\cH_{N,\mathrm{sym}}$ one has
\begin{eqnarray}
\langle\Psi,\widehat{f}\,q_1\,\Psi\rangle\;&=&\;\langle\Psi,\widehat{f}\,\widehat{m}\,\Psi\rangle \\
\langle\Psi,\widehat{f}\,q_1\,q_2\,\Psi\rangle\;&\leqslant&\;\frac{N}{N-1}\langle\Psi,\widehat{f}\,\widehat{m}^2\,\Psi\rangle\,. \label{eq:commutazione-m2}
\end{eqnarray}
\end{lemma}

% dove c'e' \label{lemma:commutation-modified} mettere {commutazione-partial-symmetry}

A further relevant tool is a modification of Lemma \ref{commutazione} above for the case when $\Psi$ carries only a partial permutation symmetry. In the present context this is the case when we consider the two-component many-body states of the form $p_1^A\Psi$ -- see the control of terms of the form $(pp,qq)_1$ in Subsection \ref{subsec:different_n_q}. We import the following result from \cite[Lemma 4.2]{Pickl-RMP-2015}:
\begin{lemma}\label{commutazione-partial-symmetry}
For any $f:\{0,\dots,N\}\rightarrow[0,+\infty)$ and any $\Phi\in\mathfrak{h}\otimes\cH_{N-1,\mathrm{sym}}$ one has
\begin{equation}\label{eq:commutation-modified}
\|\widehat{f} q_1 \Phi\|^2\;\leqslant\;\frac{N}{N-1}\,\|\widehat{f}\,\widehat{n}\,\Phi\|^2\,.
\end{equation}
\end{lemma}
In particular, in the context of Subsection \ref{subsec:V1_and_V2_terms}, the bound \eqref{eq:commutation-modified} above implies
\begin{equation}\label{eq:commutation-modified-application1}
\begin{split}
\big\langle\Psi,(\widehat{n}^{-2})^Aq_1^Aq_2^Ap_1^B\Psi\big\rangle\;&\leqslant\;\|(\widehat{n}^{-1})^A q_1^A q_2^Ap_1^B\,\Psi\|^2 \\
&\leqslant\;\frac{N_1}{N_1-1}\,\|(\widehat{n}^{-1})^A\,\widehat{n}^Aq_2^A p_1^B\,\Psi\|^2 \\
&\leqslant\;2\,\|q_2^A\Psi\|^2
\end{split}
\end{equation}
and similarly, in the context of Subsection \ref{subsec:different_n_q},
\begin{equation}\label{eq:commutation-modified-application}
\begin{split}
\langle\Psi,(\widehat{n}^{-2})^A p_2^Aq_1^Aq_1^B\,\Psi\rangle\;&=\;\|(\widehat{n}^{-1})^A q_1^A q_1^Bp_2^A\,\Psi\|^2 \\
&\leqslant\;\frac{N_1}{N_1-1}\,\|(\widehat{n}^{-1})^A\,\widehat{n}^Ap_2^A q_1^B\,\Psi\|^2 \\
&\leqslant\;2\,\|q_1^B\Psi\|^2\,.
\end{split}
\end{equation}

Next to the operators of the form $\widehat{f}$, a relevant role in the estimates for the ``counting'' method is played by the operators of the form $\widehat{\tau_n f}$, where $\tau_n$ for given $n\in\mathbb{Z}$ is the operation that produces the shifted function
\begin{equation}\label{eq:shifted_function}
(\tau_n f)(k)\;:=\;f(k+n)\,,\qquad k\in\{0,1,\dots,N\}\,.
\end{equation}
The following important property holds (see \cite[Lemma 3.10]{kp-2009-cmp2010}):
\begin{lemma}\label{lemma:exchange-hat} 
Let $A_{1\cdots r}\equiv A\otimes \mathbbm{1}_{N-r}$ be an operator on $\cH_{N}\cong\cH_{r}\otimes\cH_{N-r}$ that acts non-trivially only on the first factor $\cH_{r}$, and let $Q_j$, $j=1,2$, be two orthogonal projections on $\cH_{r}$ given by  monomials of $p$'s and $q$'s, each with $n_j$ factors $q$ and $r-n_j$ factors $p$.
Set $n:=n_2-n_1$.
Then
\begin{equation}\label{eq:commutation}
Q_1\,A_{1\dots r}\,\widehat{f} \,Q_2\;=\;Q_1\,\widehat{\tau_n f}\,A_{1\dots r}\,Q_2  
\end{equation}
as an identity of bounded operators on $\cH_N$ (with a tacit  identification $Q_j\equiv Q_j\otimes \mathbbm{1}_{N-r}$).
\end{lemma}

In fact, as a consequence of the presence of \emph{two-body} interactions only in the many-body Hamiltonian $H_{N_1,N_2}$, the use of Lemma \ref{lemma:exchange-hat} is in practice limited to the case $r=2$: formula \eqref{eq:commutation} then reads
\begin{equation}\label{eq:hat-business-for-2}
\begin{split}
p_1 \, p_2\,  A_{12}\, \widehat{f} \,q_1 \,p_2 \;&=\;p_1\,p_2\,\widehat{\tau_1 f}\,A_{12}\, q_1\,p_2\\
p_1 \, p_2\,  A_{12}\, \widehat{f} \,q_1 \,q_2 \;&=\;p_1\,p_2\,\widehat{\tau_2 f}\,A_{12}\, q_1\,q_2\\
q_1 \, p_2\,  A_{12}\, \widehat{f} \,q_1 \,q_2 \;&=\;q_1\,p_2\,\widehat{\tau_1 f}\,A_{12}\, q_1\,q_2
\end{split}
\end{equation}
while in all other cases with equal number of $q$'s on the left and on the right we have a commutation of the form $\sharp_1\sharp_2 A_{12}\,\widehat{f}\,\sharp_1\sharp_2=\sharp_1\sharp_2 \,\widehat{f}\,A_{12}\,\sharp_1\sharp_2$.

%\begin{acknowledgements}
%\end{acknowledgements}

%\bibliographystyle{spmpsci}
%\bibliography{bib_ALE.bib}


\def\cprime{$'$}
\begin{thebibliography}{10}
\providecommand{\url}[1]{{#1}}
\providecommand{\urlprefix}{URL }
\expandafter\ifx\csname urlstyle\endcsname\relax
  \providecommand{\doi}[1]{DOI~\discretionary{}{}{}#1}\else
  \providecommand{\doi}{DOI~\discretionary{}{}{}\begingroup
  \urlstyle{rm}\Url}\fi

\bibitem{Benedikter-Porta-Schlein-2015}
Benedikter, N., Porta, M., Schlein, B.: {Effective evolution equations from
  quantum dynamics}, \emph{{Springer Briefs in Mathematical Physics}}, vol.~7.
\newblock Springer, Cham (2016)

\bibitem{DeOliveira-Michelangeli-2016}
{De Oliveira}, G., Michelangeli, A.: {In preparation}  (2016)

\bibitem{kp-2009-cmp2010}
Knowles, A., Pickl, P.: {Mean-field dynamics: singular potentials and rate of
  convergence}.
\newblock Comm. Math. Phys. \textbf{298}(1), 101--138 (2010)

\bibitem{LSeSY-ober}
Lieb, E.H., Seiringer, R., Solovej, J.P., Yngvason, J.: {The mathematics of the
  {B}ose gas and its condensation}, \emph{{Oberwolfach Seminars}}, vol.~34.
\newblock Birkh{\"a}user Verlag, Basel (2005)

\bibitem{am_GPlim}
Michelangeli, A.: {Role of scaling limits in the rigorous analysis of
  {B}ose-{E}instein condensation}.
\newblock J. Math. Phys. \textbf{48}, 102,102 (2007)

\bibitem{am_equivalentBEC}
Michelangeli, A.: {Equivalent definitions of asymptotic 100\% {BEC}}.
\newblock Nuovo Cimento Sec.~B. pp. 181--192 (2008)

\bibitem{Modugno-PRL-2002}
Modugno, G., Modugno, M., Riboli, F., Roati, G., Inguscio, M.: {Two Atomic
  Species Superfluid}.
\newblock Phys. Rev. Lett. \textbf{89}, 190,404 (2002)

\bibitem{MBGCW-1997}
Myatt, C.J., Burt, E.A., Ghrist, R.W., Cornell, E.A., Wieman, C.E.: {Production
  of Two-overlapping Bose-Einstein Condensates by Sympathetic Cooling}.
\newblock Phys. Rev. Lett. \textbf{78}(4), 586--589 (1997)

\bibitem{Pachpatte-ineq}
Pachpatte, B.G.: {Inequalities for differential and integral equations},
  \emph{{Mathematics in Science and Engineering}}, vol. 197.
\newblock Academic Press, Inc., San Diego, CA (1998)

\bibitem{Pickl-JSP-2010}
Pickl, P.: {Derivation of the time dependent {G}ross-{P}itaevskii equation
  without positivity condition on the interaction}.
\newblock J. Stat. Phys. \textbf{140}(1), 76--89 (2010)

\bibitem{Pickl-LMP-2011}
Pickl, P.: {A simple derivation of mean field limits for quantum systems}.
\newblock Lett. Math. Phys. \textbf{97}(2), 151--164 (2011)

\bibitem{Pickl-RMP-2015}
Pickl, P.: {Derivation of the time dependent {G}ross-{P}itaevskii equation with
  external fields}.
\newblock Rev. Math. Phys. \textbf{27}(1), 1550,003, 45 (2015)

\bibitem{ps2003}
Pitaevskii, L., Stringari, S.: {Bose-{E}instein {C}ondensation}.
\newblock Clarendon Press, Oxford (2003)

\bibitem{S-2007}
Schlein, B.: {Dynamics of {B}ose-{E}instein Condensates} (arXiv.org:0704.0813
  (2007)).
\newblock \urlprefix\url{arXiv.org:0704.0813}

\bibitem{S-2008}
Schlein, B.: {Derivation of Effective Evolution Equations from Microscopic
  Quantum Dynamics} (arXiv.org:0807.4307 (2008))

\bibitem{Stamper-Kurn_Ketterke_et_al_PRL-1998}
Stamper-Kurn, D.M., Andrews, M.R., Chikkatur, A.P., Inouye, S., Miesner, H.J.,
  Stenger, J., Ketterle, W.: {Optical Confinement of a Bose-Einstein
  Condensate}.
\newblock Phys. Rev. Lett. \textbf{80}, 2027--2030 (1998)

\end{thebibliography}


\begin{thebibliography}{10}
\providecommand{\url}[1]{{#1}}
\providecommand{\urlprefix}{URL }
\expandafter\ifx\csname urlstyle\endcsname\relax
  \providecommand{\doi}[1]{DOI~\discretionary{}{}{}#1}\else
  \providecommand{\doi}{DOI~\discretionary{}{}{}\begingroup
  \urlstyle{rm}\Url}\fi

\bibitem{Benedikter-Porta-Schlein-2015}
Benedikter, N., Porta, M., Schlein, B.: {Effective evolution equations from
  quantum dynamics}, \emph{{Springer Briefs in Mathematical Physics}}, vol.~7.
\newblock Springer, Cham (2016)

\bibitem{cazenave}
Cazenave, T.: {Semilinear {S}chr{\"o}dinger equations}, \emph{{Courant Lecture
  Notes in Mathematics}}, vol.~10.
\newblock New York University Courant Institute of Mathematical Sciences, New
  York (2003)

\bibitem{DeOliveira-Michelangeli-2016}
{De Oliveira}, G., Michelangeli, A.: {Mean-field effective dynamics and quantum
  fluctuations for a binary condensate} (SISSA preprint 47/2016/MATE (2016))

\bibitem{Heil-2012}
Heil, T.: {Mean-field limits in bosonic systems}. \\
\newblock http://www.math.lmu.de/$\sim$bohmmech/theses/Heil\_Thomas\_MA.pdf (2012)

\bibitem{Jungel-Weishaupl2013_2compNLS_blowup}
J{\"u}ngel, A., Weish{\"a}upl, R.M.: {Blow-up in two-component nonlinear
  Schr{\"o}dinger systems with an external driven field}.
\newblock Mathematical Models and Methods in Applied Sciences \textbf{23}(09),
  1699--1727 (2013)

\bibitem{kp-2009-cmp2010}
Knowles, A., Pickl, P.: {Mean-field dynamics: singular potentials and rate of
  convergence}.
\newblock Comm. Math. Phys. \textbf{298}(1), 101--138 (2010)

\bibitem{Li-Wu-Lai-JPhysA2010-sharp_blowup_thresh_coupledNLS}
Li, X., Wu, Y., Lai, S.: {A sharp threshold of blow-up for coupled nonlinear
  Schr{\"o}dinger equations}.
\newblock Journal of Physics A: Mathematical and Theoretical \textbf{43}(16),
  165,205 (2010)

\bibitem{LSeSY-ober}
Lieb, E.H., Seiringer, R., Solovej, J.P., Yngvason, J.: {The mathematics of the
  {B}ose gas and its condensation}, \emph{{Oberwolfach Seminars}}, vol.~34.
\newblock Birkh{\"a}user Verlag, Basel (2005)

\bibitem{Ma-Zhao-JMP2008_coupledNLS}
Ma, L., Zhao, L.: {Sharp thresholds of blow-up and global existence for the
  coupled nonlinear Schr{\"o}dinger system}.
\newblock Journal of Mathematical Physics \textbf{49}(6), 062103 (2008)

\bibitem{am_GPlim}
Michelangeli, A.: {Role of scaling limits in the rigorous analysis of
  {B}ose-{E}instein condensation}.
\newblock J. Math. Phys. \textbf{48}, 102,102 (2007)

\bibitem{am_equivalentBEC}
Michelangeli, A.: {Equivalent definitions of asymptotic 100\% {BEC}}.
\newblock Nuovo Cimento Sec.~B. pp. 181--192 (2008)

\bibitem{Modugno-PRL-2002}
Modugno, G., Modugno, M., Riboli, F., Roati, G., Inguscio, M.: {Two Atomic
  Species Superfluid}.
\newblock Phys. Rev. Lett. \textbf{89}, 190,404 (2002)

\bibitem{MBGCW-1997}
Myatt, C.J., Burt, E.A., Ghrist, R.W., Cornell, E.A., Wieman, C.E.: {Production
  of Two-overlapping Bose-Einstein Condensates by Sympathetic Cooling}.
\newblock Phys. Rev. Lett. \textbf{78}(4), 586--589 (1997)

\bibitem{Pachpatte-ineq}
Pachpatte, B.G.: {Inequalities for differential and integral equations},
  \emph{{Mathematics in Science and Engineering}}, vol. 197.
\newblock Academic Press, Inc., San Diego, CA (1998)

\bibitem{Pickl-JSP-2010}
Pickl, P.: {Derivation of the time dependent {G}ross-{P}itaevskii equation
  without positivity condition on the interaction}.
\newblock J. Stat. Phys. \textbf{140}(1), 76--89 (2010)

\bibitem{Pickl-LMP-2011}
Pickl, P.: {A simple derivation of mean field limits for quantum systems}.
\newblock Lett. Math. Phys. \textbf{97}(2), 151--164 (2011)

\bibitem{Pickl-RMP-2015}
Pickl, P.: {Derivation of the time dependent {G}ross-{P}itaevskii equation with
  external fields}.
\newblock Rev. Math. Phys. \textbf{27}(1), 1550,003, 45 (2015)

\bibitem{ps2003}
Pitaevskii, L., Stringari, S.: {Bose-{E}instein {C}ondensation}.
\newblock Clarendon Press, Oxford (2003)

\bibitem{S-2007}
Schlein, B.: {Dynamics of {B}ose-{E}instein Condensates} (arXiv.org:0704.0813
  (2007)).
\newblock \urlprefix\url{arXiv.org:0704.0813}

\bibitem{S-2008}
Schlein, B.: {Derivation of Effective Evolution Equations from Microscopic
  Quantum Dynamics} (arXiv.org:0807.4307 (2008))

\bibitem{Stamper-Kurn_Ketterke_et_al_PRL-1998}
Stamper-Kurn, D.M., Andrews, M.R., Chikkatur, A.P., Inouye, S., Miesner, H.J.,
  Stenger, J., Ketterle, W.: {Optical Confinement of a Bose-Einstein
  Condensate}.
\newblock Phys. Rev. Lett. \textbf{80}, 2027--2030 (1998)

\end{thebibliography}

\def\cprime{$'$}

\end{document}